\documentclass[10pt,journal,compsoc]{IEEEtran}
\usepackage{amsmath}
\usepackage{amssymb}
\usepackage{amsthm}
\usepackage{bbm}
\usepackage{blindtext}
\usepackage{enumitem}
\usepackage{multirow}
\usepackage{graphicx}
\usepackage{color}
\theoremstyle{definition}
\newtheorem{assumption}{Assumption}
\newtheorem{theorem}{Theorem}
\newtheorem{definition}{Definition}

\newtheorem{lemma}{Lemma}

\usepackage{algorithm} 
\usepackage{algorithmic}
\usepackage{braket}

\usepackage{subfigure}
\usepackage{ulem}
\allowdisplaybreaks[4]
\begin{document}
%
\title{Byzantine-Resilient Federated Learning at Edge}
%
%
%
%

\author{Youming~Tao,~\IEEEmembership{Student Member,~IEEE}, Sijia~Cui, Wenlu~Xu, Haofei~Yin, Dongxiao~Yu,~\IEEEmembership{Senior Member,~IEEE}, Weifa~Liang~~\IEEEmembership{Senior Member,~IEEE}
        and~Xiuzhen~Cheng,~\IEEEmembership{Fellow,~IEEE}
\IEEEcompsocitemizethanks{\IEEEcompsocthanksitem Y. Tao, H. Yin, D. Yu, X. Cheng are with the School of Computer Science and Technology, Shandong University, Qingdao, Shandong, P.R. China.\protect\\
E-mail: \{ym.tao99, hf\_yin\}@mail.sdu.edu.cn, \{dxyu, xzcheng\}@sdu.edu.cn
\IEEEcompsocthanksitem S. Cui is with the Institute of Automation, University of Chinese Academy of Sciences, Beijing, P.R. China.\protect\\
E-mail:cuisijia2022@ia.ac.cn
\IEEEcompsocthanksitem W. Xu is with the Department of Statistics, University of California, Los Angeles, CA, USA.\protect\\
E-mail: wenluxu@ucla.edu
\IEEEcompsocthanksitem W. Liang is with the Department of Computer Science, City University of Hong Kong, Kowloon, Hong Kong, P.R. China.\protect\\
E-mail: weifa.liang@cityu.edu.hk
}
\thanks{Manuscript received January 23, 2022; revised October 25, 2023; accepted February 26, 2023.}}

%
%

\markboth{IEEE Transactions on Computers,~Vol.~X, No.~X, XX~2023}%
{Shell \MakeLowercase{\textit{et al.}}: Bare Demo of IEEEtran.cls for Computer Society Journals}
%



\IEEEtitleabstractindextext{%
\begin{abstract}
    Both Byzantine resilience and communication efficiency have attracted tremendous attention recently for their significance in edge federated learning. However, most existing algorithms may fail when dealing with real-world irregular data that behaves in a heavy-tailed manner. To address this issue, we study the stochastic convex and non-convex optimization problem for federated learning at edge and show how to handle heavy-tailed data while retaining the Byzantine resilience, communication efficiency and the optimal statistical error rates simultaneously. Specifically, we first present a Byzantine-resilient distributed gradient descent algorithm that can handle the heavy-tailed data and meanwhile converge under the standard assumptions. To reduce the communication overhead, we further propose another algorithm that incorporates gradient compression techniques to save communication costs during the learning process. Theoretical analysis shows that our algorithms achieve order-optimal statistical error rate in presence of Byzantine devices. Finally, we conduct extensive experiments on both synthetic and real-world datasets to verify the efficacy of our algorithms.
\end{abstract}

}

\maketitle

\IEEEdisplaynontitleabstractindextext

%
\IEEEpeerreviewmaketitle

\IEEEraisesectionheading{\section{Introduction}\label{sec:introduction}}

%
%
%
%
\IEEEPARstart{R}{ecent} years have witnessed the proliferation of smart edge devices, which leads to an unprecedented amount of data generated at the network edge. Thanks to the significant increasing in computation power of edge devices and the ubiquitous deployment of communication infrastructures, data can be processed locally and aggregated across devices efficiently. With these merits, it is natural to implement large-scale machine learning algorithms at edge, which brings about the concept of edge intelligence and has empowered many emerging applications that benefit human lives, such as smart city and autonomous driving. 

Due to the widespread concerns over data ownership and privacy, federated learning (FL), proposed by Google\cite{McMahanMRHA_AISTATS17}, has emerged as a popular paradigm for distributed ML model training, see, e.g., \cite{ZhangGQZZLA2022,FengZGQLY2022,WuHLMMJ21}. In edge FL, the data is retained in edge devices and processed in parallel, thus much more real-time results can be provided for real-world applications. However, there remain some practical issues that hinder the successful implementation of edge FL. 

One commonly encountered issue in such large-scale distributed systems arises from the potential unreliability of edge devices. Distributing the computation over multiple devices induces a higher risk of failures. In particular, some devices in the system may not follow the predefined protocol and exhibit abnormal behaviors, either actively or passively due to crashes, malfunctioning hardware, unreliable communication channels or attacks from adversaries. The inherently unpredictable behaviors of faulty devices are usually modeled as {\it Byzantine fault}~\cite{LamportSP_TOPLAS82,VeroneseCBLV13,DistlerCK16}. It has been shown in \cite{YinCRB_ICML18} that even when the number of Byzantine devices is small (even only one) and the value sent by them is moderate and even difficult to detect, the performance can still be significantly degraded. Thus, Byzantine resilience has always been a main consideration in the design of FL frameworks, see, e.g., \cite{farhadkhani2022byzantine, el2021distributed,data2020byzantine}.

Communication overhead is also an important consideration in edge FL. Typically, edge devices need to upload local results (gradients or model parameters) to the server for global aggregation repeatedly. Due to inherently limited bandwidth of wireless channels, exchanging data between edge devices and the server will incur heavy communication load and might cause network congestion, which is especially the case when the device number is huge. Heavy communication overhead is the major bottleneck that hinders the parallelism and scalability of FL at edge~\cite{shi2020communication}.

In addition, data is a key ingredient in machine learning. Recent studies, e.g., \cite{yan2020sequential,zhao2021event,zhu2020robust}, have shown that
heavy-tailed noises exist widely in practical multi-sensor systems. Since most data stored in edge devices are collected via various sensors, it is natural for the data used for training learning models to be irregular and behave in a heavy-tailed manner. Furthermore, in many real-world applications, data have been observed to be heavy-tailed in themselves, especially those from biomedicine \cite{woolson2011statistical,biswas2007statistical} and finance \cite{ibragimov2015heavy,bradley2003financial}. In a nutshell, heavy-tailed data can be widespread at edge. Heavy-tailed data could degrade the performance of learning algorithms (see, e.g., \cite{Holland_AISTATS19,holland2019efficient}), and the presence of Byzantine devices could make things worse for federated training at edge. Unfortunately, existing works on Byzantine resilience in FL all make strong assumptions on the distribution of loss gradients, for example, sub-exponential gradients~\cite{ChenSX_Sigmetrics17,ghosh2021communication}, gradients with bounded skewness~\cite{YinCRB_ICML18}, or gradients with norm-wise bounded variance~\cite{KarimireddyHJ_ICML21} (see Section~\ref{sec-relatedwork} for more details). Hence, how to mitigate the impact of heavy-tailed data on edge FL is an urgent requirement and still lack investigation.

With the perceptions above, in this paper, we consider Byzantine resilience, communication efficiency and heavy-tailed data robustness simultaneously for the first time. In particular, we have the following natural question: 

{\it Is there any way to handle the heavy-tailed data for edge federated learning while retaining the Byzantine resilience, communication efficiency, and the optimal statistical error rates?}

In this paper, we provide affirmative answer to the above question. We design an edge FL framework that is robust to heavy-tailed data as well as satisfies the requirement of Byzantine resilience and communication efficiency. Our main contributions and technical challenges are as follows.
\subsection{Main Contributions}
In the first part, we conduct a comprehensive study on the Byzantine-tolerant distributed gradient descent with heavy-tailed data under the standard assumptions. In particular, for heavy-tailed data, we assume that the distribution of loss gradients has only \textit{coordinate-wise} bounded second-order raw moment. We establish the high-probability guarantees of statistical error rate for strongly convex, general convex and non-convex population risk functions respectively. Specifically, for all the cases, we show that our algorithm achieves the following statistical error rate\footnote{Throughout this paper, the notations $\widetilde{\mathcal{O}}(\cdot)$ and $\widetilde{\Omega}(\cdot)$ hide logarithmic factors.}:
        \[
        \widetilde{\mathcal{O}}\left(d^2\left[\frac{\alpha^2}{n}+\frac{1}{mn}\right]\right),
        \]
where $\alpha\in(0,\frac{1}{2})$ is the fraction of Byzantine devices, $n$ is the size of local dataset on each edge device and $m$ is the number of edge devices. The error rate above matches the error rate given in \cite{YinCRB_ICML18} and it has been shown in \cite{YinCRB_ICML18} that, for strongly-convex population risk functions and a fixed $d$, no Byzantine-resilient algorithm can achieve an error lower than $\widetilde{\Omega}\left(\frac{\alpha^2}{n}+\frac{1}{mn}\right)$, which implies that our algorithm still achieve order-wise optimality in terms of ($\alpha$,$n$,$m$), even in the presence of heavy-tailed data.
    
 In the second part, we study how to further retain the optimal statistical error rates under the requirement of both Byzantine resilience and communication efficiency. To achieve the communication efficiency, we adopt the technique of gradient compression and consider a generic class of compressors called $\delta$-approximate compressor. Based on this, we propose a communication-efficient and Byzantine resilient distributed gradient descent algorithm with heavy-tailed data. In this case, our statistical error rates becomes:
    \[
        \widetilde{\mathcal{O}}\left(d^2\left[ \frac{\alpha^2}{n}+\frac{1-\delta}{n}+\frac{1}{mn}\right]\right),
    \]
where $\delta$ is the compression factor, and when $\delta=1$ (which implies no compression), the error rate becomes $\widetilde{\mathcal{O}}\left(\frac{\alpha^2}{n}+\frac{1}{mn}\right)$, which means that the compression term has no order-wise contribution to the error rate.

\subsection{Technical Challenges}

When only considering the Byzantine resilience, a natural and direct idea to address the problem under the heavy-tailed data setting is to replace the robust aggregator used by the server with some state-of-art robust mean estimators that can deal with the data with \textit{coordinate-wise} bounded second moment, like \cite{DongH0_NIPS19,DiakonikolasKP_NIPS20,HuR_AISTATS21}. Unfortunately, as far as we know, these robust mean estimators are not appropriate for the edge FL settings due to the following reasons. Firstly, they typically limit the corrupted data to a small fraction, which is usually not the case at edge. Secondly, to analyse the statistical error rate of the learning algorithm, it requires the estimators to have certain continuity property so that the \textit{uniform} estimation error can be bounded. The reason why the uniform error bound is needed is due to the fact that the analysis of these robust mean estimators relies on the assumption of i.i.d. data. In FL, this means the gradients computed using the training dataset should be i.i.d., which is true for a given model parameter. However, the model parameter is updated iteratively based on the training dataset and thus the parameters across the iterations are highly dependent on each other. As a result, in any iteration $t>1$, the gradients computed using the training dataset are no longer i.i.d.. Hence, we have to establish uniform concentration to bound the estimation error for all possible parameters simultaneously. It is still unclear whether these existing estimators can achieve the uniform convergence. In this paper, to solve the problem, we let the server and the devices jointly estimate the expected loss gradient in each iteration. At the device end, each device first performs a robust local mean estimator based on soft truncation and noise smoothness, which is motivated by \cite{Holland_AISTATS19}. Then the central machine aggregates the gradient estimates by coordinate-wise trimmed mean to rule out the outliers caused by Byzantine nodes. Based on this, we propose an efficient and more robust distributed gradient descent algorithm. The major challenge in analysis is to bound the uniform error when the local gradient estimator is combined with the coordinate-wise trimmed mean. We overcome this by first analysing the point-wise error bound for each coordinate and then using the coordinate-wise continuity of the local gradient estimator to obtain the uniform error bound for each coordinate via the covering arguments. 

When further considering the gradient compression, we let the server perform a norm-based trimmed mean to aggregate the compressed local estimators. The key challenge here becomes to analyse the uniform error bound for the combination of the compressed local gradient estimator and the norm-based trimmed mean. After introducing the gradient compression, the trimming process becomes norm-based and our previous "coordinate-wise" analysis no longer applies. Thus we have to adopt a different analysis. To tackle this problem, we build up on the techniques of \cite{ghosh2021communication} to directly bound the uniform error of the local estimator and then consider the impact of the compression.

\subsection{Related Work}\label{sec-relatedwork}

To cope with Byzantine attacks in distributed learning setting, most solutions rely on the outlier-robust estimators for aggregation, such as coordinate-wise median~\cite{YinCRB_ICML18}, coordinate-wise mean~\cite{YinCRB_ICML18}, geometric median~\cite{pillutla2022robust,ChenSX_Sigmetrics17} and majority voting~\cite{BernsteinWAA_ICML18,jin2020stochastic}, instead of vanilla averaging. Also, there are works developing robust aggregators by combining the ideas of these robust estimators. For example, \cite{BlanchardMGS_NIPS17} proposed Krum based on the ideas of majority voting and geometric median, and \cite{guerraoui2018hidden} proposed Bulyan that ensures majority agreement on each coordinate of the aggregated gradients by combining Krum and coordinated-wise trimmed mean. A critical issue in these approaches is that their convergence guarantees rely on strong assumptions on gradient distributions. In particular, \cite{KarimireddyHJ_ICML21} provided some examples to show that these approaches do not obtain the true optimum especially when dealing with heavy-tailed distributions. To fix it, \cite{KarimireddyHJ_ICML21} proposed to utilize worker momentum and just assumed norm-wise bounded variance for loss gradients. Inspired by this, recent work \cite{farhadkhani2022byzantine} further investigated distributed momentum for Byzantine learning under the norm-wise bounded variance assumption. However, this assumption still implies that data are well-behaved to some extent. According to recent empirical findings, e.g., \cite{simsekli2019tail}, the gradient noise could be $\alpha$-stable random variable with extremely heavy tails. In this paper, we take the first step for dealing with more extremely heavy-tailed data (or gradients) under the assumption that the loss gradients have only \textit{coordinate-wise} bounded second raw moment.

To the best of our knowledge, there are quite limited works focusing on Byzantine resilient learning and gradient compression simultaneously, except for a few notable exceptions of \cite{BernsteinZAA_ICLR19,ghosh2021communication,data2020byzantine}. \cite{BernsteinZAA_ICLR19} assumed that all devices have access to the same data and their method can only tolerate blind multiplicative adversaries (i.e., adversaries that must determine how to corrupt the gradient before observing the true gradient and can only multiply each coordinate of the true gradient by arbitrary scalar). In contrast, we consider stronger adversaries and a more general setting where different devices have different local datasets. \cite{ghosh2021communication} proposed algorithms that combined the robust aggregators and gradient compression together, and introduced error feedback to further reduce the communication costs. However, their convergence guarantees rely on the sub-exponential gradients assumption, which makes their methods not applicable to our setting. Lately, \cite{data2020byzantine} considered the Byzantine resilience and communication efficiency issues together in the heterogeneous data setting. Unfortunately, their results also rely on the norm-wise bounded variance assumption for loss gradients, while we relax this assumption in this work.

\subsection{Road Map}
The remaining part of the paper is organized as follows. The formal problem definition and system model are given in Section~\ref{sec:model}. In Section~\ref{sec-heavy-tail}, we propose a distributed gradient descent algorithm that is robust to both Byzantine fault and heavy-tailed training data. We show that our proposed algorithm can achieve the optimal statistical error rates. In Section~\ref{sec-compress}, we consider how to further reduce the communication overheads, and propose a modified algorithm by introducing the gradient compression schemes. The optimal statistical error rates are shown to be retained. In Section~\ref{sec:exp}, we report the experimental results. Finally, we conclude the paper in Section~\ref{sec:conclu}.

\vspace{-0.1in}
\section{Problem Setup and Preliminaries}\label{sec:model}

\subsection{Edge Federated Learning Problem}
We consider the stochastic convex and non-convex optimization problem. Formally, let $\mathcal{W}\subseteq\mathbb{R}^d$ be the parameter space containing all the possible model parameters and $\mathcal{D}$ be an unknown distribution over the data universe $\mathcal{Z}$. Given a loss function $\ell:\mathcal{W}\times\mathcal{Z}\to\mathbb{R}$, where $\ell(w,z)$ measures the risk induced by data $z$ under the model parameter choice $w$, and a dataset $D=\{z_1, z_2,\cdots, z_N\}$, where $z_i$'s are i.i.d. samples from the distribution $\mathcal{D}$ over $\mathcal{Z}$, the goal is to learn an optimal parameter choice $w^*\in\mathcal{W}$ that minimizes the population risk $R_{\mathcal{D}}(w)$, i.e.,
\begin{equation}\label{eq-prm}
    w^*\in\arg\min_{w\in\mathcal{W}}R_{\mathcal{D}}(w)\triangleq\mathbb{E}_{z\sim\mathcal{D}}[\ell(w,z)].\footnote{We assume that $R_{\mathcal{D}}(w)\triangleq\mathbb{E}_{z\sim\mathcal{D}}[\ell(w,z)]$ is well-defined for every $w\in\mathcal{W}$.}
\end{equation}

Note that since the data distributed $\mathcal{D}$ is unknown, the population risk function $R_{\mathcal{D}}(\cdot)$ is typically unknown in practice. Hence, we cannot compute $w^*$ straightforwardly by solving the minimization problem in (\ref{eq-prm}) with gradient descent (GD). 

We focus on solving the above stochastic optimization problem over an edge intelligent system via the federated learning framework. The edge intelligent system consists of an edge server and $m$ edge devices. We assume that the total $N$ training data are evenly distributed across the $m$ devices such that each worker machine holds $n=\frac{N}{m}$ data
\footnote{Although this is a simplified assumption on data balance over devices, our results can be easily extended to the heterogeneous data sizes setting provided the data sizes are of the same order.The same assumption has been adopted by many related works (e.g. \cite{ChenSX_Sigmetrics17,YinCRB_ICML18,ghosh2021communication}).}.

We consider a synchronous distributed system where the server can communicate with devices in each round. Among the $m$ devices, at most $\alpha$ $(\alpha<\frac{1}{2})$ fraction of devices are Byzantine and the rest $1-\alpha$ fraction are normal/good. In each round, the good devices will follow the predefined protocols faithfully. While for Byzantine ones, we have the following assumptions. Firstly, we assume that the set of Byzantine devices can be {\it dynamic} throughout the learning process. We denote the Byzantine devices in round $t$ as $\mathcal{B}_t$ and the remaining good devices as $\mathcal{G}_t$. Secondly, we assume the Byzantine devices to be {\it omniscient}, i.e., they have complete knowledge of the system and the learning algorithm, and have access to the computations made by the rest good devices. Thirdly, the Byzantine devices need not obey any predefined protocols and can send arbitrary message to the server (maybe send nothing at all) in each round. Moreover, Byzantine devices can even collude with each other. The only limit on Byzantine devices is that these devices cannot contaminate the local dataset.

\subsection{Preliminaries}\label{sec-prelim}

We first review some key concepts in optimization.

\begin{definition}[Lipschitzness]
    A function $f:\mathcal{W}\to\mathbb{R}$ is $L$-Lipschitz if for $\forall w_1, w_2\in\mathcal{W}$,
    \[
        |f(w_1)-f(w_2)|\le L\lVert w_1-w_2\rVert_2.
    \]
\end{definition}

\begin{definition}[Strong Convexity]
    A function $f$ is $a$-strongly convex on $\mathcal{W}$ if for $\forall w_1, w_2\in\mathcal{W}$, 
    \[
        f(w_1)\ge f(w_2)+\langle\nabla f(w_2),w_1-w_2\rangle+\frac{a}{2}\lVert w_1-w_2\rVert_2^2.
    \]
\end{definition}

\begin{definition}[Smoothness]
    A function $f$ is $b$-smooth on $\mathcal{W}$ if for $\forall w_1, w_2\in\mathcal{W}$, 
    \[
        f(w_1)\le f(w_2)+\langle\nabla f(w_2),w_1-w_2\rangle+\frac{b}{2}\lVert w_1-w_2\rVert_2^2.
    \]
\end{definition}

\begin{definition}[Projection]
    Given a convex set $\mathcal{W}\subseteq\mathbb{R}^d$, the projection of any $\theta\in\mathbb{R}^d$ to $\mathcal{W}$ is denoted by 
    \[
        \prod_{\mathcal{W}}\theta=\arg\min_{w\in\mathcal{W}}\lVert\theta-w\rVert.
    \]
\end{definition}

We make use of the following assumptions in the paper.

\begin{assumption}\label{A1}
	The parameter space $\mathcal{W}$ is closed, convex, and bounded with diameter $\Delta$, i.e., for $\forall w_1, w_2\in\mathcal{W}$, $\lVert w_1-w_2\rVert_2\le\Delta$.
\end{assumption}

\begin{assumption}\label{A2}
	The population risk function $R_{\mathcal{D}}(w)$ is $L_R$-smooth for $\forall w \in \mathcal{W}$, where $L_R$ is a known constant.
\end{assumption}

\begin{assumption}\label{A3}
	For any given data $z\in\mathcal{Z}$, the loss gradient $\nabla\ell(w,z)$ satisfies that for each coordinate $k\in[d]$, $\nabla_k \ell(w, z)$ is $L_k$-Lipschitz. Let $\hat{L}\triangleq\sqrt{\sum_{k=1}^{d}L_k^2}$.
\end{assumption}

The above three assumptions are quite standard and has been commonly adopted in the previous works, e.g., \cite{YinCRB_ICML18,ghosh2021communication}. Furthermore, we make the following assumption for heavy-tailed data. Specifically, for any parameter $w$, we assume that loss gradients have only coordinate-wise bounded second raw moments.

\begin{assumption}\label{A4}
	For any given $w \in \mathcal{W}$ and each coordinate $k \in [d]$, $\mathbb{E}_{z\sim\mathcal{D}}[\nabla_k^2\ell(w,x)]\le v$, where $v$ is a known constant.
\end{assumption}

We note that this assumption is reasonable and has been used in some other learning problems with heavy-tailed data as well, e.g., \cite{Holland_AISTATS19,WangX_ICML20,hu2022high,kamath2022improved}. We provide a concrete example of classical linear regression to validate this.

\setlength{\parindent}{0cm}
\paragraph*{\bf Example} Consider a linear regression model $y=\langle w^*, x\rangle+\xi$, where $x\in\mathbb{R}^{d}$ is the feature vector and $y\in\mathbb{R}$ is the label, and $\xi$ is the noise. For the noise, we assume that $\xi$ is independent of $x$, and satisfies: 1) $\mathbb{E}[\xi]=0$; 2) $\mathbb{E}[\xi^2]\le c_1$ for some constant $c_1$. For the feature vector, we assume that the coordinates of $x=(x_1\cdots,x_d)$ are independent of each other, and $x$ satisfies: 1) $\mathbb{E}[x]=0$; 2) $\mathbb{E}[\|x\|_2^2]\le c_2$ for some constant $c_2$; 3) for $\forall k\in[d]$, $\mathbb{E}[x_k^4]\le c_3$ for some constant $c_3$. We consider the quadratic loss function $\ell(w,(x,y))=\frac{1}{2}(y-\langle w, x\rangle)^2$, then $\nabla_k\ell(w,(x,y))=(y-\langle w,x\rangle)x_k$. By some simple computation, we have for all $w\in\mathcal{W}$ that,
\begin{flalign*}
    &\quad\mathbb{E}[\nabla_k^2\ell(w,(x,y))]\\
    &= \mathbb{E}[(y-\langle w,x\rangle)^2x_k^2]= \mathbb{E}[(\langle w^*-w,x\rangle+\xi)^2x_k^2]\\
    &= \mathbb{E}[\langle w^*-w,x\rangle^2x_k^2]+\mathbb{E}[\xi^2]\cdot\mathbb{E}[x_k^2]\\
    &\le \lVert w^*-w\rVert^2\mathbb{E}[x_k^2\cdot\lVert x\rVert_2^2]+\mathbb{E}[\xi^2]\cdot\mathbb{E}[x_k^2]\\
    &\le \Delta^2\mathbb{E}[x_k^2\cdot\lVert x\rVert_2^2]+\mathbb{E}[\xi^2]\cdot\mathbb{E}[x_k^2]< \Delta^2(c_3+c_2^2)+c_1\cdot c_2
\end{flalign*}
where the third equality is because $\xi$ is independent of $x$ and $\mathbb{E}[\xi]=0$, and the second inequality is due to Assumption~\ref{A1}. It can be seen from above that the loss gradient only has coordinate-wise bounded second moment, which validates Assumption~\ref{A4}. Furthermore, denote the bound we got by $C$. Then the norm-wise variance of the loss gradient in our example is bounded as $\mathbb{E}[\|\nabla\ell(w,(x,y))\|_2^2]\le d\cdot C$, which is proportional to the dimension $d$. In contrast, bounded norm-wise variance assumption used in previous works only assumed it to be a universal constant, which indicates that our Assumption~\ref{A4} is weaker and thus more general.

\section{Byzantine-Resilient Heavy-tailed Gradient Descent}\label{sec-heavy-tail}

In this section, we study how to handle the heavy-tailed data while retaining the Byzantine resilience and the statistical error rate. We propose a \underline{B}yzantine-resilient \underline{h}eavy-tailed \underline{g}radient \underline{d}escent algorithm called BHGD.

\subsection{Algorithm Design}

 Since $R_{\mathcal{D}}(\cdot)$ is unknown, it is infeasible to apply gradient descent algorithm directly due to the impossibility to compute the exact population risk gradient $\nabla R_{\mathcal{D}}(\cdot)$. A natural alternative way is to estimate $\nabla R_{\mathcal{D}}(\cdot)$ using the data samples $D=\{z_1,z_2,\cdots,z_N\}$. In our algorithm, the estimation is done jointly by devices and the central server. Each device first calculates a local estimate of $\nabla R_{\mathcal{D}}(\cdot)$ (for simplicity, we denote the local estimator of device $i\in[m]$ by $g_i(\cdot)$) from its local dataset and sends the local estimation to the server. The server aggregates the received $\{g_1, g_2,\cdots,g_m\}$ by coordinate-wise trimmed mean and then updates the parameter vector (see Algorithm~\ref{Algo1} for details).
 
\begin{algorithm}[tbp]
    \caption{Byzantine-Resilient Heavy-tailed Gradient Descent (BHGD)}
    \label{Algo1}
    \begin{algorithmic}[1]
        \REQUIRE Initial parameter vector $w_0\in\mathcal{W}$, step size $\eta$, time horizon $T$.
		\ENSURE $\zeta\gets\frac{1}{(\Delta n\widehat{L})^d(m+1)d(mn)^d}, s \gets \sqrt{\frac{nv}{2\log(1/\zeta)}}, \tau\gets\sqrt{2\log(1/\zeta)}$.
		\FOR{$t\gets 0,1,\cdots,T-1$}
		    \STATE \underline{Server:} Send $w_t$ to all the worker machines
		    \STATE Each \underline{good device} $i\in\mathcal{G}_t$ do in parallel: \begin{enumerate}
		        \item Computes local estimate of gradient $g_i(w_t)$. Specifically, for $\forall k\in[d]$,
		            \begin{align*}
		                g_{i,k}(w_t) \gets &\frac{1}{n}\sum_{j=1}^{n}\left[g\left(1-\frac{g^2}{2s^2\tau}\right)-\frac{g^3}{6s^2}\right]\\
		                &+\frac{s}{n}\sum_{j=1}^{n}C\left(\frac{g}{s},\frac{|g|}{s\sqrt{\tau}}\right),
		            \end{align*}
		            where $g$ denotes $\nabla_k\ell(w_t,x_j)$.
		            
		            \item Sends $g_{i}(w_t)$ to the central machine
		    \end{enumerate}
		    \STATE \underline{Server:}
		    \begin{enumerate}
		        \item For each $k\in[d]$:
		        \begin{itemize}
		            \item Sorts $g_{i,k}(w_t)$'s in a non-decreasing order.
 		        \item Removes the largest and smallest $\beta$ fraction of elements in $\{g_{i,k}(w_t)\}_{i=1}^{m}$ and denotes the indices of the remaining elements as $\mathcal{U}_{k,t}$
 		        \item Aggregates by $$g_k(w_t)=\frac{1}{|\mathcal{U}_{k,t}|}\sum_{i\in\mathcal{U}_{k,t}}g_{i,k}(w_t)$$
 		        and denotes $g(w_t)=(g_1(w_t),\cdots,g_d(w_t))$.
		        \end{itemize}
		        \item Updates the parameter by  $$w_{t+1}=\prod_{\mathcal{W}}\left(w_t-\eta\cdot g(w_t)\right)$$
		    \end{enumerate}
		\ENDFOR
    \end{algorithmic}
\end{algorithm}

The local gradient estimator in our algorithm is inspired by the robust mean estimator for heavy-tailed distribution given in \cite{Holland_AISTATS19}. To be self-contained, we first review the estimator. For simplicity, we consider a one-dimensional random variable $x\sim\mathcal{X}$ and assume that $x_1, x_2,\cdots,x_n$ are i.i.d. samples of $x$. The robust estimator consists of three steps:

\begin{enumerate}
    \item {\textbf{Scaling and Truncation} }
        For each sample $x_i$, we re-scale it by dividing $s$ and apply a soft truncation function $\phi$ on the re-scaled one. Then we calculate the empirical mean of the altered samples and put the mean back to the original scale. That is,
        \begin{equation*}
            \frac{s}{n}\sum_{i=1}^{n}\phi(\frac{x_i}{s})\approx\mathbb{E}[x].
        \end{equation*}
    \item {\textbf{Noise Multiplication} }
        Let $\epsilon_1,\cdots,\epsilon_2$ be independent random noise generated from a common distribution $\nu$ with $\mathbb{E}[\epsilon_i]=0$ for each. We multiply each sample $x_i$ by $(1+\epsilon_i)$, and then perform the scaling and truncation step on $x_i\cdot(1+\epsilon_i)$. That is,
        \begin{equation*}
            \widetilde{x}(\epsilon)=\frac{s}{n}\sum_{i=1}^{n}\phi(\frac{x_i+\epsilon_i x_i}{s}).
        \end{equation*}
    \item {\textbf{Noise Smoothing} }
        We smooth the multiplicative noise via taking the expectation with respect to the noise distribution $\nu$. That is,
        \begin{equation}\label{eq-estimator}
            \widehat{x}=\mathbb{E}[\widetilde{x}(\epsilon)]=\frac{s}{n}\sum_{i=1}^n\int\phi(\frac{x_i+\epsilon_i x_i}{s})d\nu(\epsilon_i).
        \end{equation}
\end{enumerate}

We note that the randomness in final estimator in (\ref{eq-estimator}) is only dependent on the original samples. The explicit form of the integral in (\ref{eq-estimator}) is related to the choice of the soft truncation function $\phi(\cdot)$ and the noise distribution $\nu$. The results in \cite{Catoni_Arxiv17} shows that, if we set $\phi$ to be 
\begin{equation}\label{eq-funcphi}
    \phi(x)=
    \begin{cases}
        \frac{2\sqrt{2}}{3}, &x>\sqrt{2}\\
        x-\frac{x^3}{6}, &-\sqrt{2}\le x\le \sqrt{2}\\
        -\frac{2\sqrt{2}}{3}, &x<-\sqrt{2}
    \end{cases}
\end{equation}
and set $\nu=\mathcal{N}(0,\frac{1}{\tau})$, then the integral in (\ref{eq-estimator}) has an explicit form such that it can be computed efficiently. Generally, for any $a$ and $b$, we have
\begin{equation}\label{eq-genest}
    \mathbb{E}_\nu[\phi(a+b\sqrt{\tau}\epsilon)]=a(1-\frac{b^2}{2})-\frac{a^3}{6}+C(a,|b|).
\end{equation}
The term $C(a,|b|)$ in (\ref{eq-genest}) is a correction term with a simple form. To give its explicit form, We first define some preparatory notations:

\begin{flalign*}
    &V_-\triangleq\frac{\sqrt{2}-a}{|b|}, 
    \quad\quad V_+\triangleq\frac{\sqrt{2}+a}{|b|},\\
    &F_-\triangleq\Phi(-V_-), \quad\quad F_+\triangleq\Phi(-V_+),\\
    &E_-\triangleq\exp(-\frac{V_-^2}{2}), \quad E_+\triangleq\exp(-\frac{V_+^2}{2}),
\end{flalign*}
 where $\Phi$ denotes the CDF of the standard Gaussian distribution. 
 Then with these atomic elements, the explicit form of $C(a,|b|)$ can be described as follows:
 \begin{equation*}
     C(a,|b|)=T_1+T_2+T_3+T_4+T_5,
 \end{equation*}
 where
 \begin{flalign*}
 &T_1\triangleq\frac{2\sqrt{2}}{3}(F_--F_+)\\
 &T_2\triangleq-(a-\frac{a^3}{6})(F_-+F_+)\\
 &T_3\triangleq\frac{|b|}{\sqrt{2\pi}}(1-\frac{a^2}{2})(E_+-E_-)\\
 &T_4\triangleq\frac{ab^2}{2}(F_++F_-+\frac{1}{\sqrt{2\pi}}(V_+E_++V_-E_-))\\
 &T_5\triangleq\frac{|b|^3}{6\sqrt{2\pi}}((2+V_-^2)E_--(2+V_+^2)E_+).
 \end{flalign*}

The main idea of Algorithm~\ref{Algo1} is that, instead of using empirical mean as the local estimator which may be subject to the heavy-tailed outliers, we let each device apply the one-dimensional robust mean estimator described above to each coordinate of its local loss gradients so that a more accurate local estimator $g_i(\cdot)$ for $\nabla R_{\mathcal{D}}(\cdot)$ can be obtained. Specifically, in our setting, the parameter $a, b$ in (\ref{eq-genest}) should be $\frac{\nabla_k\ell(w_t,x_j)}{s}, \frac{\nabla_k\ell(w_t,x_j)}{s\sqrt{\tau}}$ respectively, and the final estimator is described in step $3(1)$ in Algorithm~\ref{Algo1}. The server then uses the coordinate-wise trimmed mean to aggregate these local estimators and obtain a global estimator $g(\cdot)$ for $\nabla R_{\mathcal{D}}(\cdot)$ (step $4(1)$). Note that, since the trimming threshold $\beta$ is at least $\alpha$, the trimming operation ensures that the effect of Byzantine devices can be removed and hence the global estimator $g(\cdot)$ is close to $\nabla R_{\mathcal{D}}(\cdot)$.

\subsection{Theoretical Results}

In this part, we analyse the performance of Algorithm~\ref{Algo1}. Specifically, we study the statistical error rates for strongly convex, general convex and non-convex population risk function respectively. For strongly-convex and general-convex case, we focus on the excess population risk, i.e., $R_{\mathcal{D}}(w_T)-R_{\mathcal{D}}(w^*)$. For non-convex case, we focus on the rate of convergence to a critical point of the population risk, i.e., $\min_{t=0,1,\cdots,T}\lVert\nabla R_{\mathcal{D}}(w_t)\rVert_2$. For all the cases, we provide the high probability upper bounds on these error rates.

The analysis of the error rates relies on the gradient estimation error in each iteration $t\in[T]$, i.e., $\|g(w_t)-\nabla R_{\mathcal{D}}(w_t)\|$. Hence the key step is to bound the uniform error of $g(w)$ for all $w\in\mathcal{W}$. Specifically, we have the following high probability bound on $\lVert g(w)-\nabla R_{\mathcal{D}}(w)\rVert_2$ for all $w\in\mathcal{W}$.

\begin{lemma}\label{le-unierr}
For all $w\in\mathcal{W}$, it holds with probability at least $1-\frac{1}{(mn)^d}$ that
\begin{equation*}
    \lVert g(w)-\nabla R_{\mathcal{D}}(w)\rVert_2 \in \mathcal{O}\left(\alpha d\sqrt{\frac{\log(mn)}{n}}+d\sqrt{\frac{\log(mn)}{mn}}\right).
\end{equation*}
\end{lemma}

We now provide statistical error rates for our algorithm.

\noindent{\bf Strongly convex population risks:} We first consider the case where the population risk function $R_{\mathcal{D}}(\cdot)$ is strongly convex. Note that the loss function for each data $\ell(\cdot, z)$ need not be strongly convex.

\begin{theorem}\label{th-hsco}
    Suppose Assumption \ref{A1}, \ref{A2}, \ref{A3} and \ref{A4} hold, and $R_{\mathcal{D}}(\cdot)$ is $\lambda_R$-strongly convex. Choose step size $\eta=\frac{1}{L_R}$ and run Algorithm~\ref{Algo1} for $T$ rounds, then with probability at least $1-\frac{1}{(mn)^d}$, we have the following bound on excess population risk,
    \begin{flalign*}
        &R_{\mathcal{D}}(w_T)-R_{\mathcal{D}}(w^*)\\
        &\le L_R(1-\frac{\lambda_R}{L_R+\lambda_R})^{2T}\lVert w_0-w^*\rVert_2^2+\frac{4L_R}{\lambda_R^2}\mathcal{E}^2,
    \end{flalign*}
    where $\mathcal{E}\in\mathcal{O}\left(\alpha d\sqrt{\frac{\log(mn)}{n}}+d\sqrt{\frac{\log(mn)}{mn}}\right)$.
\end{theorem}

\noindent{\bf General convex population risks:} For the general convex population risk case, we need a mild technical assumption on the size of the parameter space $\mathcal{W}$.

\begin{assumption}\label{A5}
    The parameter space $\mathcal{W}$ contains the following $\ell_2$ ball centered at $w^*$:
    \[
        \{w\in\mathbb{R}^d:\lVert w-w^*\rVert_2\le2\lVert w_0-w^*\rVert_2\}.
    \]
\end{assumption}

Then we have the following result on excess population risk function.

\begin{theorem}\label{th-hgc}
    Suppose Assumption \ref{A1}, \ref{A2}, \ref{A3}, \ref{A4} and \ref{A5} hold, and $R_{\mathcal{D}}(\cdot)$ is convex. Choose step size $\eta=\frac{1}{L_R}$ and run Algorithm~\ref{Algo1} for $T=\frac{L_R}{\mathcal{E}}\lVert w_0-w^*\rVert_2$ rounds, then with probability at least $1-\frac{1}{(mn)^d}$, we have the following bound on excess population risk,
    \begin{equation*}
        R_{\mathcal{D}}(w_T)-R_{\mathcal{D}}(w^*)\le16\Delta\mathcal{E}+\frac{1}{2L_R}\mathcal{E}^2,
    \end{equation*}
    where $\mathcal{E}\in\mathcal{O}\left(\alpha d\sqrt{\frac{\log(mn)}{n}}+d\sqrt{\frac{\log(mn)}{mn}}\right)$.
\end{theorem}

\noindent{\bf Non-convex population risks:} For the non-convex population risk case, we need a slightly distinct technical assumption on the size of $\mathcal{W}$.

\begin{assumption}\label{A6}
    Suppose that for all $w\in\mathcal{W}$, $\lVert\nabla R_{\mathcal{D}}(w)\rVert_2\le G$. We assume that $\mathcal{W}$ contains the $\ell_2$ ball centered at the initial parameter $w_0$: 
    \[
        \{w\in\mathbb{R}^d:\lVert w-w_0\rVert_2\le 2\frac{G+\mathcal{E}}{\mathcal{E}^2}[R_{\mathcal{D}}(w_0)-R_{\mathcal{D}}(w^*)]\}.
    \]
\end{assumption}

We have the following guarantee on the rate of convergence to a critical point of the population risk $R_{\mathcal{D}}(\cdot)$.

\begin{theorem}\label{th-hnc}
    Suppose Assumption \ref{A1}, \ref{A2}, \ref{A3}, \ref{A4} and \ref{A6} hold. Choose step size $\eta=\frac{1}{L_R}$ and run Algorithm~\ref{Algo1} for $T=\frac{2L_R}{\mathcal{E}^2}[R_{\mathcal{D}}(w_0)-R_{\mathcal{D}}(w^*)]$ rounds, then with probability at least $1-\frac{1}{(mn)^d}$, we have
    \begin{equation*}
        \min_{t=0,\cdots,T}\lVert \nabla R_{\mathcal{D}}(w_t)\rVert_2\le\sqrt{2}\mathcal{E},
    \end{equation*}
    where $\mathcal{E}\in\mathcal{O}\left(\alpha d\sqrt{\frac{\log(mn)}{n}}+d\sqrt{\frac{\log(mn)}{mn}}\right)$.
\end{theorem}

 It can be seen from the results above that, for all the cases, our algorithm achieves an error rate of $\widetilde{\mathcal{O}}\left(d^2\left[\frac{\alpha^2}{n}+\frac{1}{mn}\right]\right)$. The error rate we obtain matches the error rates given in \cite{YinCRB_ICML18}. It has been shown in \cite{YinCRB_ICML18} that, for strongly-convex population risk functions and a fixed $d$, no algorithm can achieve an error lower than $\widetilde{\Omega}\left(\frac{\alpha^2}{n}+\frac{1}{mn}\right)$, which implies that our algorithm is order-wise optimal in terms of ($\alpha$,$n$,$m$) in this case, even when considering the heavy-tailed data.
 
\subsection{Analysis of Algorithm~\ref{Algo1}}\label{sec-appendA}

{\bf Notation:}
 Recall that we denote the Byzantine devices and the good devices in round $t$ as $\mathcal{B}_t$ and $\mathcal{G}_t$ respectively. Moreover, we use $\mathcal{U}_{k,t}$ and $\mathcal{T}_{k,t}$ to denote the untrimmed devices and the trimmed devices with respect to coordinate $k$ in round $t$. For notation simplicity, we will drop the subscript $t$ when the context is clear. 

\subsubsection{Proof of Lemma~\ref{le-unierr}}

To prove Lemma~\ref{le-unierr}, we require the following two lemmas related to local estimators $g_{i}(\cdot)$'s. The following two Lemmas show that $g_{i,k}(w)$ is concentrated around $\nabla_k R_{\mathcal{D}}(w)$ for all good devices $i\in\mathcal{G}$, any fixed coordinate $k\in[d]$ and any fixed parameter $w\in\mathcal{W}$.

\begin{lemma}\label{le-maxerr}
    For any fixed $w\in\mathcal{W}$ and coordinate $k\in[d]$, the following holds with probability at least $1-m\cdot\zeta$,
    \begin{equation}\label{eq-maxerrdw}
        \max_{i\in\mathcal{G}}| g_{i,k}(w)-\nabla_k R_{\mathcal{D}}(w)|\le\sqrt{\frac{2v\log(1/\zeta)}{n}}+\sqrt{\frac{v}{n}}.
    \end{equation}
\end{lemma}
\begin{lemma}\label{le-meanerr}
    For any fixed $w\in\mathcal{W}$ and coordinate $k\in[d]$, the following holds with probability at least $1-\zeta$,
    \begin{equation}\label{eq-meanerrdw}
       \left|\frac{1}{|\mathcal{G}|}\sum_{i\in\mathcal{G}}g_{i,k}(w)-\nabla_k R_{\mathcal{D}}(w)\right|\le \sqrt{\frac{2v\log(1/\zeta)}{(1-\alpha)mn}}+\sqrt{\frac{v}{(1-\alpha)mn}}.
    \end{equation}
\end{lemma}

\begin{proof}[Proof of Lemma~\ref{le-maxerr}]
    The one-dimension estimator defined in (\ref{eq-estimator}) has the following pointwise accuracy, which is given in \cite{Holland_AISTATS19}:
    \begin{lemma}{\cite[Lemma 2]{Holland_AISTATS19}}\label{le-poiacc}
        Consider the dataset $\{x_i\}_{i=1}^n$ where $x_i$ are i.i.d. samples drawn from distribution $\mathcal{X}$. Assume that $\mathcal{X}$ has finite second-order moment and $\mathbb{E}_\mathcal{X}[|x|^2]\le v$. Then with probability at least $1-\zeta$, the estimator $\widehat{x}$ defined in (\ref{eq-estimator}) using truncation function defined in (\ref{eq-funcphi}), noise distribution $\nu=\mathcal{N}(0,\frac{1}{\tau})$ and scale $s=\sqrt{\frac{nv}{2\log(1/\zeta)}}$ satisfies
        \begin{equation*}
            |\widehat{x}-\mathbb{E}_\mathcal{X}[x]|
            \le\sqrt{\frac{2v\log(1/\zeta)}{n}}+\sqrt{\frac{v}{n}}.
        \end{equation*}
    \end{lemma}
    According to Lemma~\ref{le-poiacc}, we have for any $i\in\mathcal{G}$, $k\in[d]$ and $w\in\mathcal{W}$ that, with probability at least $1-\zeta$,
        \begin{equation*}
            |g_{i,k}(w)-\nabla_k R_{\mathcal{D}}(w)|\le\sqrt{\frac{2v\log(1/\zeta)}{n}}+\sqrt{\frac{v}{n}}.
        \end{equation*}
    Then by taking union bound over all good devices $i\in\mathcal{G}$, we have with probability at least $1-m\cdot\zeta$ that,
    \begin{equation*}
        \max_{i\in\mathcal{G}}|g_{i,k}(w)-\nabla_k R_{\mathcal{D}}(w)|\le\sqrt{\frac{2v\log(1/\zeta)}{n}}+\sqrt{\frac{v}{n}},
    \end{equation*}
    which concludes the proof.
\end{proof}

\begin{proof}[Proof of Lemma~\ref{le-meanerr}]
    By Lemma~\ref{le-poiacc}, we have with probability at least $1-\zeta$ that,
    \begin{equation*}
        \left|\frac{1}{|\mathcal{G}|}\sum_{i\in\mathcal{G}}g_{i,k}(w)-\nabla_k R_{\mathcal{D}}(w)\right|\le\sqrt{\frac{2v\log(1/\zeta)}{|\mathcal{G}|n}}+\sqrt{\frac{v}{|\mathcal{G}|n}}.
    \end{equation*}
    Since $|\mathcal{G}|\ge(1-\alpha)m$, we obtain with probability at least $1-\zeta$ that
    \begin{equation*}
        \left|\frac{1}{|\mathcal{G}|}\sum_{i\in\mathcal{G}}g_{i,k}(w)-\nabla_k R_{\mathcal{D}}(w)\right|\le\sqrt{\frac{2v\log(1/\zeta)}{(1-\alpha)mn}}+\sqrt{\frac{v}{(1-\alpha)mn}},
    \end{equation*}
    which concludes the proof.
\end{proof}

With Lemma~\ref{le-maxerr} and Lemma~\ref{le-meanerr}, we are now ready to prove Lemma~\ref{le-unierr}.

    Due to Lemma~\ref{le-maxerr} and Lemma~\ref{le-meanerr}, we already have (\ref{eq-maxerrdw}) and (\ref{eq-meanerrdw}) hold for any fixed $w\in\mathcal{W}$ and $k\in[d]$. Next, to extend the pointwise accuracy to the uniform accuracy that holds for all $w\in\mathcal{W}$, we need to utilize the standard covering net argument. Let $\mathcal{W}_\epsilon=\{w^1, w^2,\cdots,w^{N_\epsilon}\}$ be a finite subset of $\mathcal{W}$ such that for any $w\in\mathcal{W}$, there exists some $w^p\in\mathcal{W}_\epsilon$ satisfying $\lVert w-w^p\rVert_2\le\epsilon$. According to the basic property of covering numbers for compact subsets of Euclidean space \cite{Tikhomirov1993}, we know that $N_\epsilon\le(\frac{3\Delta}{2\epsilon})^d$. Take the union bound, we have both (\ref{eq-maxerrdw}) and (\ref{eq-meanerrdw}) hold for any $k\in[d]$ and all $w=w^p\in\mathcal{W}_\epsilon$ with probability at least $1-N_\epsilon(m+1)\zeta$. 
    
    Then consider an arbitrary $w\in\mathcal{W}$. Suppose that $\lVert w-w^p\rVert_2\le\epsilon$. Since in Assumption~\ref{A3}, we assume that for each $k\in[d]$, $\nabla_k\ell(w,z)$ is $L_k$-Lipschitz for all data $z$, we know that
    \begin{equation}\label{eq-lipR}
        |\nabla_k R_{\mathcal{D}}(w)-\nabla_k R_{\mathcal{D}}(w^p)|\le L_k\epsilon.
    \end{equation}
    According to \cite[Lemma 4]{Holland_AISTATS19}, the one-dimension estimator defined in (\ref{eq-estimator}) satisfies that $|\widehat{x}(X)-\widehat{x}(X^\prime)|\le\frac{c_\nu}{n}\lVert X-X^\prime\rVert_1$ where $X,X^\prime$ denote two datasets and $c_\nu$ is a constant that equals $1-2\Phi(-\sqrt{\tau})+\sqrt{\frac{2}{\tau\pi}}\exp(-\frac{\tau}{2})$. For each coordinate $k\in[d]$ we have
    \begin{flalign}\label{eq-lipest}
        |g_{i,k}(w)-g_{i,k}(w^p)|
        &\le\frac{c_\nu}{n}\sum_{j=1}^n\left|\nabla_k\ell(w,z_j)-\nabla_k\ell(w^p,z_j)\right|\notag\\
        &\le\frac{c_\nu}{n}\cdot n\cdot L_k\lVert w-w^p \rVert_2\le c_\nu L_k\epsilon,
    \end{flalign}
    where the second inequality is due to Assumption \ref{A3}.
    Based on (\ref{eq-lipR}) and (\ref{eq-lipest}), we obtain for any $k\in[d]$ and all $w\in\mathcal{W}$ that, with probability at least $1-N_\epsilon(m+1)\zeta$,
    \begin{equation}\label{eq-maxerrd}
        \max_{\mathcal{G}}|g_{i,k}(w)-\nabla_k R_{\mathcal{D}}(w)|
        \le\sqrt{\frac{2v\log(1/\zeta)}{n}}+\sqrt{\frac{v}{n}}+(1+c_\nu)L_k\epsilon,
    \end{equation}
    and 
    \begin{flalign}\label{eq-meanerrd}
        &\left|\frac{1}{|\mathcal{G}|}\sum_{i\in\mathcal{G}}g_{i,k}(w)-\nabla_k R_{\mathcal{D}}(w)\right|\notag\\
        &\le \sqrt{\frac{2v\log(1/\zeta)}{(1-\alpha)mn}}+\sqrt{\frac{v}{(1-\alpha)mn}}+(1+c_\nu)L_k\epsilon.
    \end{flalign}

    We next move on to $g(\cdot)$. We have the following for all $w\in\mathcal{W}$ and any coordinate $k\in[d]$.
    \begin{flalign}
        &\quad\left|g_k(w)-\nabla_k R_{\mathcal{D}}(w)\right|\notag\\
        &=\left|\frac{1}{|\mathcal{U}_k|}\sum_{i\in\mathcal{U}_k}(g_{i,k}(w)-\nabla_k R_{\mathcal{D}}(w))\right|\notag\\
        &\le\frac{1}{|\mathcal{U}_k|}\left|\sum_{i\in\mathcal{G}}(g_{i,k}(w)-\nabla_k R_{\mathcal{D}}(w))\right|\notag\\
        &\quad+ \frac{1}{|\mathcal{U}_k|}\left|\sum_{i\in\mathcal{G}\cap\mathcal{T}}(g_{i,k}(w)-\nabla_k R_{\mathcal{D}}(w))\right|\notag\\
        &\quad+\frac{1}{|\mathcal{U}_k|}\left|\sum_{i\in\mathcal{B}\cap\mathcal{U}_k}(g_{i,k}(w)-\nabla_k R_{\mathcal{D}}(w))\right|.\label{eq-errterm}
    \end{flalign}
    We bound each term in (\ref{eq-errterm}) respectively. By (\ref{eq-meanerrd}) we have
    \begin{flalign}\label{eq-gterm1}
        &\frac{1}{|\mathcal{U}_k|}\left|\sum_{i\in\mathcal{G}}(g_{i,k}(w)-\nabla_k R_{\mathcal{D}}(w))\right|\notag\\
        &=\frac{|\mathcal{G}|}{|\mathcal{U}_k|}\left|\frac{1}{|\mathcal{G}|}\sum_{i\in\mathcal{G}}(g_{i,k}(w)-\nabla_k R_{\mathcal{D}}(w))\right|\notag\\
        &\le \frac{1-\alpha}{1-2\beta}\left(\sqrt{\frac{2v\log(1/\zeta)}{(1-\alpha)mn}}+\sqrt{\frac{v}{(1-\alpha)mn}}+(1+c_\nu)L_k\epsilon\right).
    \end{flalign}
    By (\ref{eq-maxerrd}), we have
    \begin{flalign}\label{eq-gterm2}
        &\frac{1}{|\mathcal{U}_k|}\left|\sum_{i\in\mathcal{G}\cap\mathcal{T}}(g_{i,k}(w)-\nabla_k R_{\mathcal{D}}(w))\right|\notag\\
        &\le\frac{2\beta}{1-2\beta}\max_{i\in\mathcal{G}}| g_{i,k}(w)-\nabla_k R_{\mathcal{D}}(w)|\notag\\
        &\le \frac{2\beta}{1-2\beta}\left(\sqrt{\frac{2v\log(1/\zeta)}{n}}+\sqrt{\frac{v}{n}}+(1+c_\nu)L_k\epsilon\right).
    \end{flalign}
     Since $\beta\ge\alpha$, w.l.o.g., we assume that $\mathcal{G}\cap\mathcal{T}_k\neq\emptyset$. Then by (\ref{eq-maxerrd}) again, we have
    \begin{flalign}\label{eq-gterm3}
        &\frac{1}{|\mathcal{U}_k|}\left|\sum_{i\in\mathcal{B}\cap\mathcal{U}_k}(g_i(w)-\nabla R_{\mathcal{D}}(w))\right|\notag\\
        &\le\frac{\alpha}{1-2\beta}\max_{i\in\mathcal{G}}|g_{i,k}(w)-\nabla_k R_{\mathcal{D}}(w)|\notag\\
        &\le \frac{\alpha}{1-2\beta}\left(\sqrt{\frac{2v\log(1/\zeta)}{n}}+\sqrt{\frac{v}{n}}+(1+c_\nu)L_k\epsilon\right).
    \end{flalign}
    It is worth noting that, all (\ref{eq-gterm1}), (\ref{eq-gterm2}) and (\ref{eq-gterm3}) hold as long as both (\ref{eq-maxerrd}) and (\ref{eq-meanerrd}) hold, which is with probability at least $1-N_\epsilon(m+1)\zeta$. Hence, With the same probability, we have the following for all $w\in\mathcal{W}$ and any $k\in[d]$:
    \begin{flalign*}
        &|g_k(w)-\nabla_k R_{\mathcal{D}}(w)|\\
        &\le \frac{\alpha+2\beta}{1-2\beta}\left(\sqrt{\frac{2v\log(1/\zeta)}{n}}+\sqrt{\frac{v}{n}}+(1+c_\nu)L_k\epsilon\right)\notag\\
        &\quad+\frac{1-\alpha}{1-2\beta}\left(\sqrt{\frac{2v\log(1/\zeta)}{(1-\alpha)mn}}+\sqrt{\frac{v}{(1-\alpha)mn}}+(1+c_\nu)L_k\epsilon\right).
    \end{flalign*}
    Taking the union bound for all $k\in[d]$ yields
    \begin{flalign}
        &\lVert g(w)-\nabla R_{\mathcal{D}}(w)\rVert_2\notag\\
        &\le \sqrt{2}\cdot\Bigg(\frac{\alpha+2\beta}{1-2\beta}\sqrt{\frac{2vd\log(1/\zeta)}{n}}+\frac{\alpha+2\beta}{1-2\beta}\sqrt{\frac{vd}{n}}\notag\\
        &\quad+\frac{1-\alpha}{1-2\beta}\sqrt{\frac{2vd\log(1/\zeta)}{(1-\alpha)mn}}+\frac{1-\alpha}{1-2\beta}\sqrt{\frac{vd}{(1-\alpha)mn}}\notag\\
        &\quad+\frac{1+2\beta}{1-2\beta}(1+c_\nu)\widehat{L}\epsilon\Bigg),
    \end{flalign}
    which holds for all $w\in\mathcal{W}$ with probability at least $1-N_\epsilon(m+1)d\zeta$.
    Setting $\epsilon=\frac{3}{2n\widehat{L}}$ and noting that $N_\epsilon=(\Delta n\widehat{L})^d$, $\zeta=\frac{1}{(\Delta n\widehat{L})^d(m+1)d(mn)^d}$, we obtain for all $w\in\mathcal{W}$ that, with probability at least $1-\frac{1}{(mn)^d}$,
    \begin{equation}
        \lVert g(w)-\nabla R_{\mathcal{D}}(w)\rVert_2\in \mathcal{O}\left(\alpha d\sqrt{\frac{\log(mn)}{n}}+d\sqrt{\frac{\log(mn)}{mn}}\right).
    \end{equation}

\subsubsection{Proof of Theorem~\ref{th-hsco}}\label{proof-th1}


Recall that in round $t$, Algorithm~\ref{Algo1} updates the parameter vector by 
\begin{equation*}
    w_{t+1}=\prod_\mathcal{W}(w_t-\eta\cdot g(w_t)).
\end{equation*} 
Define $\widehat{w}_{t+1}=w_t-\eta\cdot g(w_t)$ and thus we have $w_{t+1}=\prod_\mathcal{W}\widehat{w}_{t+1}$. By the property of Euclidean projection, we know that
\begin{equation*}
    \lVert w_{t+1}-w^*\rVert_2\le\lVert \widehat{w}_{t+1}-w^*\rVert_2.
\end{equation*}
Hence we can further obtain
\begin{flalign}\label{eq-itergap}
    &\| w_{t+1}-w^*\|_2\le\| w_t-\eta g(w_t)-w^*\|_2\notag\\
    &\le\| w_t-\eta\nabla R_{\mathcal{D}}(w_t)-w^*\|_2+\eta\| g(w_t)-\nabla R_{\mathcal{D}}(w_t)\|_2.
\end{flalign}
By the co-coercivity of strongly convex function (\cite[Lemma 3.11]{Bubeck2017convex}), we have
\begin{flalign*}
    &\langle w_t-w^*,\nabla R_{\mathcal{D}}(w_t)\rangle\notag\\
    &\ge\frac{L_R\lambda_R}{L_R+\lambda_R}\lVert w_t-w^*\rVert_2^2+\frac{1}{L_R+\lambda_R}\lVert \nabla R_{\mathcal{D}}(w_t)\rVert_2^2.
\end{flalign*}
Taking $\eta=\frac{1}{L_R}$, for the first term on the right hand side of (\ref{eq-itergap}), we have
\begin{flalign*}
    &\lVert w_t-\eta\nabla R_{\mathcal{D}}(w_t)-w^*\rVert_2^2\\
    &= \lVert w_t-w^*\rVert_2^2-2\eta\langle w_t-w^*,\nabla R_{\mathcal{D}}(w_t)\rangle+\eta^2\lVert\nabla R_{\mathcal{D}}(w_t)\rVert_2^2\\
    &\le(1-\frac{2\lambda_R}{L_R+\lambda_R})\lVert w_t-w^*\rVert_2^2-\frac{2}{L_R(L_R+\lambda_R)}\lVert\nabla R_{\mathcal{D}}(w_t)\rVert_2^2\\
    &\quad+\frac{1}{L_R^2}\lVert\nabla R_{\mathcal{D}}(w_t)\rVert_2^2\\
    &\le(1-\frac{2\lambda_R}{L_R+\lambda_R})\lVert w_t-w^*\rVert_2^2,
\end{flalign*}
where the second inequality is due to the fact that $\lambda_R<L_R$. Using the fact that $\sqrt{1-x}\le1-\frac{x}{2}$, we get
\begin{equation}\label{eq-sct1}
    \lVert w_t-\eta\nabla R_{\mathcal{D}}(w_t)-w^*\rVert_2\le(1-\frac{\lambda_R}{L_R+\lambda_R})\lVert w_t-w^*\rVert_2.
\end{equation}
Combining (\ref{eq-itergap}) and (\ref{eq-sct1}), we get
\begin{equation}\label{eq-iter}
    \lVert w_{t+1}-w^*\rVert_2\le(1-\frac{\lambda_R}{L_R+\lambda_R})\lVert w_t-w^*\rVert_2+\frac{1}{L_R}\mathcal{E},
\end{equation}
where $\mathcal{E}$ denotes the uniform upper bound on $\lVert g(w_t)-\nabla R_{\mathcal{D}}(w_t)\rVert$ and $\mathcal{E}\in\mathcal{O}\left(\alpha d\sqrt{\frac{\log(mn)}{n}}+d\sqrt{\frac{\log(mn)}{mn}}\right)$ according to Lemma~\ref{le-unierr}.
By iterating (\ref{eq-iter}) and simplifying the geometric series, we get
\begin{equation*}
    \lVert w_T-w^*\rVert_2\le(1-\frac{\lambda_R}{L_R+\lambda_R})^{T}\lVert w_0-w^*\rVert_2+\frac{2}{\lambda_R}\mathcal{E}.
\end{equation*}
By the smoothness of $R_{\mathcal{D}}(\cdot)$, we have 
\begin{flalign}
    &R_{\mathcal{D}}(w_T)-R_{\mathcal{D}}(w^*)\notag\\
    &\le\frac{L_R}{2}\lVert w_T-w^* \rVert_2^2\notag\\
    &\le\frac{L_R}{2}\left((1-\frac{\lambda_R}{L_R+\lambda_R})^{T}\lVert w_0-w^*\rVert_2+\frac{2}{\lambda_R}\mathcal{E}\right)^2\notag\\
    &\le L_R(1-\frac{\lambda_R}{L_R+\lambda_R})^{2T}\lVert w_0-w^*\rVert_2^2+\frac{4L_R}{\lambda_R^2}\mathcal{E}^2,
\end{flalign}
which concludes the proof.

\subsubsection{Proof of Theorem~\ref{th-hgc}}\label{proof-th2}


We first show that with Assumption~\ref{A5} and the choice of step size $\eta=\frac{1}{L_R}$, the iterative parameters (i.e., $w_t$'s) all stay in $\mathcal{W}$ without using projection. According to the update rule, we have

\begin{flalign}\label{eq-itergap2}
    &\lVert w_{t+1}-w^*\rVert_2\notag\\
    &\le\lVert w_t-\eta\nabla R_{\mathcal{D}}(w_t)-w^* \rVert_2+\eta\lVert g(w_t)-\nabla R_{\mathcal{D}}(w_t) \rVert_2.
\end{flalign}

For the first term on the right hand side of the inequality above, we have
\begin{flalign*}
    &\lVert w_t-\eta\nabla R_{\mathcal{D}}(w_t)-w^*\rVert_2^2\\
    &= \lVert w_t-w^*\rVert_2^2-2\eta\langle \nabla R_{\mathcal{D}}(w_t), w_t-w^*\rangle+\eta^2\lVert \nabla R_{\mathcal{D}}(w_t)\rVert_2^2\notag\\
    &\le\lVert w_t-w^* \rVert_2^2-2\eta\frac{1}{L_R}\lVert\nabla R_{\mathcal{D}}(w_t) \rVert_2^2+\eta^2\lVert \nabla R_{\mathcal{D}}(w_t)\rVert_2^2\notag\\
    &=\lVert w_t-w^*\rVert_2^2-\frac{1}{L_R^2}\lVert \nabla R_{\mathcal{D}}(w_t)\rVert_2^2\le\lVert w_t-w^*\rVert_2^2,
\end{flalign*}
where the first inequality is due to the co-coercivity of convex functions. Taking this back to (\ref{eq-itergap2}), we get
\begin{equation}\label{eq-itergapupdt}
    \lVert w_{t+1}-w^*\rVert_2\le\lVert w_t-w^*\rVert_2+\frac{\mathcal{E}}{L_R},
\end{equation}
Define $D_t\triangleq\lVert w_0-w^*\rVert_2+\frac{t\mathcal{E}}{L_R}$ for $t=0,\cdots,T$. Then we have $\|w_{t}-w^*\|_2\le D_t$ and $D_0=\lVert w_0-w^*\rVert_2$. Since $T=\frac{L_R D_0}{\mathcal{E}}$, we can conclude that for all $t=0,\cdots,T-1$, $w_t\in\mathcal{W}$. 

By the smoothness of $R_{\mathcal{D}}(w_t)$, we have
\begin{flalign}
    &R_{\mathcal{D}}(w_{t+1})\notag\\
    &\le R_{\mathcal{D}}(w_t)+\langle \nabla R_{\mathcal{D}}(w_t), w_{t+1}-w_{t}\rangle+\frac{L_R}{2}\lVert w_{t+1}-w_{t}\rVert_2^2\notag\\
    &= R_{\mathcal{D}}(w_t)+\eta\langle \nabla R_{\mathcal{D}}(w_t), -g(w_t)+\nabla R_{\mathcal{D}}(w_t)-\nabla R_{\mathcal{D}}(w_t)\rangle\notag\\
    &\quad+\eta^2\frac{L_R}{2}\lVert g(w_t)\notag-\nabla R_{\mathcal{D}}(w_t)+\nabla R_{\mathcal{D}}(w_t)\rVert_2^2\notag\\
    &= R_{\mathcal{D}}(w_t)-\eta\langle\nabla R_{\mathcal{D}}(w_t), g(w_t)-\nabla R_{\mathcal{D}}(w_t)\rangle\notag\\
    &-\eta\lVert \nabla R_{\mathcal{D}}(w_t)\rVert_2^2+\eta^2\frac{L_R}{2}\lVert g(w_t)-\nabla R_{\mathcal{D}}(w_t)\rVert_2^2\notag\\
    &+\eta^2\frac{L_R}{2}\lVert \nabla R_{\mathcal{D}}(w_t)\rVert_2^2\notag\\
    &+\eta^2 L_R\langle g(w_t)-\nabla R_{\mathcal{D}}(w_t), \nabla R_{\mathcal{D}}(w_t) \rangle\notag\\
    &\le R_{\mathcal{D}}(w_t)-\frac{1}{2L_R}\lVert \nabla R_{\mathcal{D}}(w_t) \rVert_2^2+\frac{1}{2L_R}\mathcal{E}^2.\label{eq-cvxrskupd}
\end{flalign}

The remainder of the proof relies on the following lemma. 
\begin{lemma}\label{le-cvxaux}
    When $R_{\mathcal{D}}(w)$ is convex and $g(\cdot)$ satisfies Lemma~\ref{le-unierr} for all $w\in\mathcal{W}$, by running Algorithm~\ref{Algo1} for $T=\frac{L_R D_0}{\mathcal{E}}$ rounds, there exists some $t\in\{0,\cdots,T\}$ such that $R_{\mathcal{D}}(w_{t})-R_{\mathcal{D}}(w^*)\le 16 D_0\mathcal{E}$.
\end{lemma}

\begin{proof}[Proof of Lemma~\ref{le-cvxaux}]
    Since $T=\frac{L_R D_0}{\mathcal{E}}$, we have $D_t \le 2D_0$ for all $t = 0,1,\cdots, T$. According to the first order optimality of convex functions, for any $w$:
	\begin{flalign}
		&R_{\mathcal{D}}(w)-R_{\mathcal{D}}(w^*) \notag\\
		&\le \Braket{\nabla R_{\mathcal{D}}(w), w-w^*} \le \Vert \nabla R_{\mathcal{D}}(w)\Vert_2\Vert w-w^*\Vert_2,
	\end{flalign}
		thus we have
		\begin{equation}\label{eq-lbnormR}
			\Vert \nabla R_{\mathcal{D}}(w) \Vert_2 \ge \frac{R_{\mathcal{D}}(w)-R_{\mathcal{D}}(w^*)}{\Vert w-w^*\Vert_2}.
		\end{equation}
	Suppose that there exists $t \in  \{0,1,\cdots, T-1\}$ s.t. $\Vert \nabla R_{\mathcal{D}}(w_t)\Vert_2 < \sqrt{2}\mathcal{E}$, then we have:
	\[
		R_{\mathcal{D}}(w_t)-R_{\mathcal{D}}(w^*) \le \Vert \nabla R_{\mathcal{D}}(w_t)\Vert_2\Vert w_t-w^*\Vert_2\le 2\sqrt{2}D_0\mathcal{E}. 
	\]
	Otherwise, $\Vert \nabla R_{\mathcal{D}}(w_t)\Vert_2 \ge \sqrt{2}\mathcal{E}$ holds for all $t \in  \{0,1,\cdots, T-1\}$. Then for all $t\in\{0,\cdots, T-1\}$, we have the following:
	\[
		\begin{split}
			&R_{\mathcal{D}}(w_{t+1})-R_{\mathcal{D}}(w^*)\\
			&\le R_{\mathcal{D}}(w_t)-R_{\mathcal{D}}(w^*)-\frac{1}{2L_R}\Vert \nabla R_{\mathcal{D}}(w_t) \Vert_2^2 + \frac{1}{2L_R}\mathcal{E}^2\\
			&\le R_{\mathcal{D}}(w_t) - R_{\mathcal{D}}(w^*) - \frac{1}{4L_R} \Vert \nabla R_{\mathcal{D}}(w_t)\Vert_2^2\\
			&\le R_{\mathcal{D}}(w_t) - R_{\mathcal{D}}(w^*) - \frac{1}{4L_R D_t^2}(R_{\mathcal{D}}(w_t)-R_{\mathcal{D}}(w^*))^2.\\
		\end{split}
	\]
	Multiplying both sides by $[(R_{\mathcal{D}}(w_{t+1})-R_{\mathcal{D}}(w^*))(R_{\mathcal{D}}(w_{t})-R_{\mathcal{D}}(w^*))]^{-1}$ and rearranging the terms, we obtain:
		\begin{flalign}\label{eq-imp}
		    &\frac{1}{R_{\mathcal{D}}(w_{t+1})-R_{\mathcal{D}}(w^*)}\notag\\
		    &\ge \frac{1}{R_{\mathcal{D}}(w_{t})-R_{\mathcal{D}}(w^*)} + \frac{1}{4L_R D_t^2}\frac{R_{\mathcal{D}}(w_{t})-R_{\mathcal{D}}(w^*)}{R_{\mathcal{D}}(w_{t+1})-R_{\mathcal{D}}(w^*)} \notag\\
		    &\ge \frac{1}{R_{\mathcal{D}}(w_{t})-R_{\mathcal{D}}(w^*)} + \frac{1}{16L_R D_0^2},
		\end{flalign}
		where the last inequality comes from $\frac{R_{\mathcal{D}}(w_{t})-R_{\mathcal{D}}(w^*)}{R_{\mathcal{D}}(w_{t+1})-R_{\mathcal{D}}(w^*)} \ge 1 \ge \frac{D_t^2}{4D_0^2}$.
		
	Note that the inequality (\ref{eq-imp}) implies
	\[
		    \frac{1}{R_{\mathcal{D}}(w_{T})-R_{\mathcal{D}}(w^*)} \ge \frac{T}{16L_R D_0^2},
	\]
		thus we obtain $R_{\mathcal{D}}(w_T)-R_{\mathcal{D}}(w^*) \le 16D_0\mathcal{E}$ using the fact $T=\frac{L_R D_0}{\mathcal{E}}.$
\end{proof}

With Lemma~\ref{le-cvxaux}, we next show that $R_{\mathcal{D}}(w_T)-R_{\mathcal{D}}(w^*)\le 16D_0\mathcal{E}+\frac{1}{2L_R}\mathcal{E}^2$. In particular, let $t=t_0$ be the first time that $R_{\mathcal{D}}(w_t)-R_{\mathcal{D}}(w^*)\le 16D_0\mathcal{E}$, we show that we have $R_{\mathcal{D}}(w_t)-R_{\mathcal{D}}(w^*)\le16D_0\mathcal{E}+\frac{1}{2L_R}\mathcal{E}^2$ hold for any $t>t_0$. If this is false, then we let $t_1>t_0$ be the first time that $R_{\mathcal{D}}(w_{t})-R_{\mathcal{D}}(w^*)>16D_0\mathcal{E}+\frac{1}{2L_R}\mathcal{E}^2$. Then there must be $R_{\mathcal{D}}(w_{t_1-1})<R_{\mathcal{D}}(w_{t_1})$. According to (\ref{eq-cvxrskupd}), the following should also hold
\begin{equation*}
    \begin{split}
        &R_{\mathcal{D}}(w_{t_1-1})-R_{\mathcal{D}}(w^*)\\
        &\ge R_{\mathcal{D}}(w_{t_1})-R_{\mathcal{D}}(w^*)-\frac{1}{2L_R}\mathcal{E}^2>16D_0\mathcal{E}.
    \end{split}
\end{equation*}
Then according to (\ref{eq-lbnormR}) and Assumption~\ref{A5}, we have
\begin{equation*}
    \lVert \nabla R_{\mathcal{D}}(w_{t_1-1})\rVert_2\ge\frac{R_{\mathcal{D}}(w_{t_1-1})-R_{\mathcal{D}}(w^*)}{\lVert w_{t_1-1}-w^*\rVert_2}>8\mathcal{E}.
\end{equation*}
Taking this into (\ref{eq-cvxrskupd}), we get $R_{\mathcal{D}}(w_{t_1})\le R_{\mathcal{D}}(w_{t_1-1})-\frac{63}{2L_R}\mathcal{E}^2$, implying that $R_{\mathcal{D}}(w_{t_1})\le R_{\mathcal{D}}(w_{t_1-1})$, which contradicts with $R_{\mathcal{D}}(w_{t_1-1})<R_{\mathcal{D}}(w_{t_1})$. Hence for any $t>t_0$, $R_{\mathcal{D}}(w_t)-R_{\mathcal{D}}(w^*)\le16D_0\mathcal{E}+\frac{1}{2L_R}\mathcal{E}^2$, which concludes the proof.

\subsubsection{Proof of Theorem~\ref{th-hnc}}\label{proof-th3}


We first show that with Assumption~\ref{A6} and the choice of step size $\eta=\frac{1}{L_R}$, the iterative parameters (i.e., $w_t$'s) all stay in $\mathcal{W}$ without using projection. By the update rule, we have
\begin{flalign*}
    &\| w_{t+1}-w_0\|_2\\
    &\le\lVert w_t-w_0\rVert_2+\eta(\lVert\nabla R_{\mathcal{D}}(w_t)\rVert_2+\lVert g(w_t)-\nabla R_{\mathcal{D}}(w_t)\rVert_2)\notag\\
    &\le\lVert w_t-w_0\rVert_2+\frac{1}{L_R}(G+\mathcal{E}).
\end{flalign*}
Since we set $T=\frac{L_R}{\mathcal{E}^2}(R_{\mathcal{D}}(w_0)-R_{\mathcal{D}}(w^*))$, according to Assumption~\ref{A6}, we know that $w_t \in \mathcal{W}$ for $t = 0,1,\cdots, T$ without projection.
By the smoothness of $R_{\mathcal{D}}(\cdot)$ and the choice of $\eta=\frac{1}{L_R}$, the inequality (\ref{eq-cvxrskupd}) still holds i.e., for all $t = 0,1,\cdots, T$,
\begin{equation}\label{eq2-cvxrskupd}
    R_{\mathcal{D}}(w_{t+1})\le R_{\mathcal{D}}(w_t)-\frac{1}{2L_R}\lVert \nabla R_{\mathcal{D}}(w_t) \rVert_2^2+\frac{1}{2L_R}\mathcal{E}^2. 
\end{equation}
Sum up (\ref{eq2-cvxrskupd}) for $t = 0,1,\cdots, T-1$. Then, we get 
\[
    \begin{split}
        0 &\le R_{\mathcal{D}}(w_{T}) - R_{\mathcal{D}}(w^*)\\
        &\le R_{\mathcal{D}}(w_0)- R_{\mathcal{D}}(w^*)-\frac{1}{2L_R}\sum_{t = 0}^{T-1}\lVert \nabla R_{\mathcal{D}}(w_t) \rVert_2^2+\frac{T}{2L_R}\mathcal{E}^2,
    \end{split}
\]
which implies that
\[
    \min_{t = 0,1,\cdots, T-1}\lVert \nabla R_{\mathcal{D}}(w_t) \rVert_2^2 \le \frac{2L_R}{T}(R_{\mathcal{D}}(w_0) - R_{\mathcal{D}}(w^*))+\mathcal{E}^2.
\]
Note that $T=\frac{2L_R}{\mathcal{E}^2}[R_{\mathcal{D}}(w_0)-R_{\mathcal{D}}(w^*)]$, we conclude the proof.

\section{Byzantine-Resilient Heavy-tailed Gradient Descent with Compression}\label{sec-compress}

In this section, we study how to further reduce the communication overhead in addition to retaining the Byzantine resilience and the statistical error rate. Based on BHGD, we introduce gradient compression and propose \underline{B}yzantine-resilient \underline{h}eavy-tailed \underline{g}radient \underline{d}escent with \underline{c}ompression, which is called BHGD-C .

\subsection{Algorithm Design}

\begin{algorithm}[tbp]
    \caption{Byzantine-Resilient Heavy-tailed Gradient Descent with Compression (BHGD-C)}
    \label{Algo2}
    \begin{algorithmic}[1]
        \REQUIRE Initial parameter vector $w_0\in\mathcal{W}$, compressor $\mathcal{Q}(.)$, step size $\eta$, time horizon $T$.
		\ENSURE $\zeta \gets \frac{1}{2(\Delta\sqrt{mn})^dd(mn)^d} $, $s \gets \sqrt{\frac{nv}{2\log(1/\zeta)}}, \tau\gets\sqrt{2\log(1/\zeta)}$.
		\FOR{$t\gets 0,1,\cdots,T-1$}
		    \STATE \underline{Server:} Send $w_t$ to all the worker machines
		    \STATE Each \underline{good device} $i\in\mathcal{G}_t$ do in parallel: \begin{enumerate}
		        \item Computes local estimate of gradient $g_i(w_t)$. Specifically, for $\forall k\in[d]$,
		            \begin{align*}
		                g_{i,k}(w_t)\gets
		                & \frac{1}{n}\sum_{j=1}^{n}\left[g\left(1-\frac{g^2}{2s^2\tau}\right)-\frac{g^3}{6s^2}\right]\\
		                &+\frac{s}{n}\sum_{j=1}^{n}C\left(\frac{g}{s},\frac{|g|}{s\sqrt{\tau}}\right),
		            \end{align*}
		            where $g$ denotes $\nabla_k\ell(w_t,x_j)$.
		            
		            \item Sends $\mathcal{Q}(g_{i}(w_t))$ to the central machine.
		    \end{enumerate}
		    \STATE \underline{Server:}
		    \begin{enumerate}
		        \item Sorts $\mathcal{Q}(g_i(w_t))$'s in a non-decreasing order according to $\lVert \mathcal{Q}(g_{i}(w_t))\rVert_2$
		        \item Denotes the indices of the first $1-\beta$ fraction of elements as $\mathcal{U}_t$
		        \item Aggregates the gradients through the trimmed mean: $g(w_t)=\frac{1}{|\mathcal{U}_t|}\sum_{i\in\mathcal{U}_t}\mathcal{Q}(g_{i}(w_t))$
		        \item Updates the parameter by  $$w_{t+1}=\prod_{\mathcal{W}}\left(w_t-\eta\cdot g(w_t)\right)$$
		    \end{enumerate}
		\ENDFOR
    \end{algorithmic}
\end{algorithm}

To reduce the communication cost, we adopt the technique of gradient compression. For the compression scheme, we consider a generic class of compression operators called $\delta$-approximate compressors, just as the recent work \cite{ghosh2021communication} did. The formal definition for the compressors is given below.

\begin{definition}[$\delta$-Approximate Compressor]
An operator $\mathcal{Q}(\cdot):\mathbb{R}^d\to\mathbb{R}^d$ is said to be an $\delta$-approximate compressor on a set $\mathcal{S}\subseteq\mathbb{R}^d$ if $\forall x\in\mathcal{S}$,
\begin{equation*}
    \lVert \mathcal{Q}(x)-x\rVert_2^2\le(1-\delta)\lVert x\rVert_2^2,
\end{equation*}
where $\delta\in(0,1]$ is the compression factor. 
\end{definition}

The compression factor $\delta$ measures the degree of compression and $\delta=1$ implies $\mathcal{Q}(x)=x$, which means no compression. There are many compressors satisfying the definition, such as Top-$k$ Sparsification \cite{StichCJ_NIPS18}, 
 $k$-PCA \cite{WangSLCPW_NIPS18}, 
Randomized Quantization \cite{AlistarhG0TV_NIPS17}, 
 $1$-bit Quantization \cite{BernsteinWAA_ICML18},
$\ell_1$-norm Quantization \cite{KarimireddyRSJ_ICML19}, etc. 

In Algorithm~\ref{Algo2}, we let each non-Byzantine device compress its estimate for loss gradient by a $\delta$-approximate compressor $\mathcal{Q}(\cdot)$ before sending it to the server (step 3(2)). No restriction is placed on Byzantine devices. Note that, in Algorithm~\ref{Algo2}, the aggregation rule used by the server is different from that of Algorithm~\ref{Algo1}. Now the server performs a norm-based trimmed mean (i.e., to trim the gradients according their norm values, see step 4(1)-4(2)). By doing this, the server eliminates only $\beta$ $(\beta\ge\alpha)$ fraction of local gradient estimators instead of $2\beta$ as in Algorithm~\ref{Algo1}. And we believe this will bring a more accurate estimation for the server.

\subsection{Theoretical Results}

In this part, we analyse the influence of the gradient compression on the learning performance. Throughout the analysis, we need an additional mild assumption on population risk function $R_{\mathcal{D}}(\cdot)$.

\begin{assumption} \label{A-Gbound}
    For all $w \in \mathcal{W}$, $\Vert \nabla R_{\mathcal{D}}(w)\Vert_2 \le G$, where $G$ is a constant.
\end{assumption}

Note that, while the loss gradients are unbound, it is realistic to assume that the expected gradient (i.e., the population risk gradient) is inside a ball with some radius $G$. 

Before presenting the statistical error rates, We first analyse the uniform accuracy of $g(w)$ for all $w\in\mathcal{W}$.

\begin{lemma}\label{lm-T_0}
With Assumption~\ref{A-Gbound}, for all $w\in\mathcal{W}$, it holds with probability at least $1-\frac{1}{(mn)^d}$ that
\begin{flalign}\label{bd-T_0}
    &\| g(w)-\nabla R_{\mathcal{D}}(w)\|_2\notag\\ &\in\mathcal{O}\Big((\alpha+\sqrt{1-\delta})d\sqrt{\frac{\log(mn)}{n}}+d\sqrt{\frac{\log(mn)}{mn}}\Big).
\end{flalign}
\end{lemma}

Next, we provide the main results on the error rates .

\noindent{\bf Strongly convex population risks:} Note that the loss function for each data $\ell(\cdot, z)$ need not be strongly convex. The upper bound on the excess population risk is as follows.

\begin{theorem}\label{Thm-hsco}
    Suppose Assumption \ref{A1}, \ref{A2}, \ref{A3}, \ref{A4} and \ref{A-Gbound} hold, and $R_{\mathcal{D}}(\cdot)$ is $\lambda_R$-strongly convex. Choose step size $\eta=\frac{1}{L_R}$ and run Algorithm~\ref{Algo2} for $T$ rounds, then with probability at least $1-\frac{1}{(mn)^d}$, we have the following bound on excess population risk,
    \begin{flalign*}
        &R_{\mathcal{D}}(w_T)-R_{\mathcal{D}}(w^*)\\
        &\le L_R(1-\frac{\lambda_R}{L_R+\lambda_R})^{2T}\lVert w_0-w^*\rVert_2^2+\frac{4L_R}{\lambda_R^2}\widetilde{\mathcal{E}}^2,
    \end{flalign*}
    where $\widetilde{\mathcal{E}}\in \mathcal{O}\left((\alpha+\sqrt{1-\delta})d\sqrt{\frac{\log(mn)}{n}}+d\sqrt{\frac{\log(mn)}{mn}}\right)$.
\end{theorem}

\noindent{\bf General convex population risks:} We have the following upper bound on excess population risk function.

\begin{theorem}\label{Thm-hgc}
    Suppose Assumption \ref{A1}, \ref{A2}, \ref{A3}, \ref{A4},  \ref{A5} and \ref{A-Gbound} hold, and $R_{\mathcal{D}}(\cdot)$ is convex. Choose step size $\eta=\frac{1}{L_R}$ and run Algorithm~\ref{Algo1} for $T=\frac{L_R}{\widetilde{\mathcal{E}}}\lVert w_0-w^*\rVert_2$ rounds, then with probability at least $1-\frac{1}{(mn)^d}$, we have the following bound on excess population risk,
    \begin{equation*}
        R_{\mathcal{D}}(w_T)-R_{\mathcal{D}}(w^*)\le16\Delta\widetilde{\mathcal{E}}+\frac{1}{2L_R}\widetilde{\mathcal{E}}^2,
    \end{equation*}
    where $\widetilde{\mathcal{E}}\in\mathcal{O}\left((\alpha+\sqrt{1-\delta})d\sqrt{\frac{\log(mn)}{n}}+d\sqrt{\frac{\log(mn)}{mn}}\right).$
\end{theorem}

\noindent{\bf Non-convex population risks:} For the non-convex population risk case, we need a slightly distinct technical assumption on the size of $\mathcal{W}$.

\begin{assumption}\label{A7+1}
    The parameter space $\mathcal{W}$ contains the $\ell_2$-ball centered at $w_0$:
    \[\{
    w  \in \mathcal{R}^d: \Vert w-w_0\Vert_2 \le 2\frac{G+\widetilde{\mathcal{E}}}{\widetilde{\mathcal{E}}^2}(R_{\mathcal{D}}(w_0)-R_{\mathcal{D}}(w^*))
    \}.\]
\end{assumption}

\begin{theorem}\label{Thm-com-non}
	Suppose Assumption \ref{A1}, \ref{A2}, \ref{A3}, \ref{A4}, \ref{A-Gbound} and \ref{A7+1} hold, and $R_{\mathcal{D}}(\cdot)$ is non-convex. Choose step size $\eta = \frac{1}{L_R}$ and run Algorithm~\ref{Algo2} for $T=\frac{2L_R}{\widetilde{\mathcal{E}}^2}(R_{\mathcal{D}}(w_0)-R_{\mathcal{D}}(w^*))$ rounds, then with probability at least $1-\frac{1}{(mn)^d}$, we have
	\[
	    \min_{t=0,1,\cdots, T}\Vert \nabla R_{\mathcal{D}}(w_t)\Vert_2^2 \le \sqrt{2}\widetilde{\mathcal{E}},
	\]
	where $\widetilde{\mathcal{E}}\in\mathcal{O}\left((\alpha+\sqrt{1-\delta})d\sqrt{\frac{\log(mn)}{n}}+d\sqrt{\frac{\log(mn)}{mn}}\right).$
\end{theorem}

\subsection{Analysis of Algorithm~\ref{Algo2}}

{\bf Notation:}
We denote the Byzantine devices and the good devices in round $t$ as $\mathcal{B}_t$ and $\mathcal{G}_t$ respectively. Also, we use $\mathcal{U}_{t}$ and $\mathcal{T}_{t}$ to denote the untrimmed devices and the trimmed devices in round $t$.

The analysis of algorithm~\ref{Algo2} takes Lemma~\ref{lm-T_0} as the core. To prove Lemma~\ref{lm-T_0}, we require the following two lemmas, which show that the local estimators $g_{i}(\cdot)$'s are concentrated around $\nabla R_{\mathcal{D}}(w)$ for all $w\in\mathcal{W}$.

\begin{lemma}\label{le-maxunierr}
    For all $w\in\mathcal{W}$, the following holds with probability at least $1-m\cdot(\Delta\sqrt{n})^dd\zeta$,
    \begin{flalign}\label{eq-maxunierr}
            &\max_{i\in\mathcal{G}}\|g_i(w)-\nabla R_{\mathcal{D}}(w)\|_2\notag\\
            &\le\frac{3(c_v\widehat{L}+L_R)}{2\sqrt{n}}+\sqrt{\frac{2vd\log(\zeta^{-1})}{n}}+\sqrt{\frac{v}{n}}\triangleq\mathcal{E}_1
    \end{flalign}
\end{lemma}

\begin{lemma}\label{le-avgunierr}
    For all $w\in\mathcal{W}$, the following holds with probabillity at least $1-(\Delta\sqrt{mn})^dd\zeta$,
    \begin{flalign}\label{eq-avgunierr}
            &\left\|\frac{1}{|\mathcal{G}|}\sum_{i\in\mathcal{G}}g_i(w)-g(w)\right\|_2\notag\\
            &\le\frac{3(c_v\widehat{L}+L_R)}{2\sqrt{mn}}+\sqrt{\frac{2vd\log(\zeta^{-1})}{mn}}+\sqrt{\frac{v}{mn}}\triangleq\mathcal{E}_2
    \end{flalign}
\end{lemma}

\begin{proof}[Proof of Lemma~\ref{le-maxunierr}]
    According to \cite[Lemma 5]{Holland_AISTATS19}, we know that for any fixed $i\in\mathcal{G}$ and all $w\in\mathcal{W}$, the following holds with probability at least $1-N_\epsilon d\zeta$,
    \[
        \|g_i(w)-\nabla R_{\mathcal{D}}(w)\|_2\le (c_v\widehat{L}+L_R)\epsilon+\sqrt{\frac{2vd\log(\zeta^{-1})}{n}}+\sqrt{\frac{v}{n}},
    \]
    where $N_\epsilon\le(\frac{3\Delta}{2\epsilon})^d$. By setting $\epsilon=\frac{3}{2\sqrt{n}}$, we obtain that with probability at least $1-(\Delta\sqrt{n})^dd\zeta$,
    \[
        \|g_i(w)-\nabla R\|_2\le\frac{3(c_v\widehat{L}+L_R)}{2\sqrt{n}}+\sqrt{\frac{2vd\log(\zeta^{-1})}{n}}+\sqrt{\frac{v}{n}}.
    \]
    Take the union bound, we know that with probability at least $1-m\cdot(\Delta\sqrt{n})^dd\zeta$, the following holds for all $w\in\mathcal{W}$,
    \[
        \begin{split}
            &\max_{i\in\mathcal{G}}\left\|g_i(w)-\nabla R_{\mathcal{D}}(w)\right\|_2\\
            &\le\frac{3(c_v\widehat{L}+L_R)}{2\sqrt{n}}+\sqrt{\frac{2vd\log(\zeta^{-1})}{n}}+\sqrt{\frac{v}{n}}.
        \end{split}
    \]
\end{proof}

\begin{proof}[Proof of Lemma~\ref{le-avgunierr}]
    Following the same argument in the proof of Lemma~\ref{le-maxunierr} (except for taking $\epsilon=\frac{3}{2\sqrt{mn}}$), we have that, with probability at least $1-(\Delta\sqrt{mn})^dd\zeta$, the following holds for all $w\in\mathcal{W}$,
    \[
        \begin{split}
            &\left\|\frac{1}{|\mathcal{G}|}\sum_{i\in\mathcal{G}}g_{i}(w)-\nabla R_{\mathcal{D}}(w)\right\|_2\\
            &\le\frac{3(c_v\widehat{L}+L_R)}{2\sqrt{mn}}+\sqrt{\frac{2vd\log(\zeta^{-1})}{mn}}+\sqrt{\frac{v}{mn}}.
        \end{split}
    \]
\end{proof}
With Lemma~\ref{le-maxunierr} and Lemma~\ref{le-avgunierr}, we are now ready to prove Lemma~\ref{lm-T_0}. For notation simplicity, we denote $\|g(w)-\nabla R_{\mathcal{D}}(w)\|_2$ by $\widetilde{\mathcal{E}}(w)$. First, we take union bound, then with probability at least $1-m\cdot(\Delta\sqrt{n})^dd\zeta-(\Delta\sqrt{mn})^dd\zeta\ge1-2(\Delta\sqrt{mn})^dd\zeta$, (\ref{eq-maxunierr}) and (\ref{eq-avgunierr}) hold simultaneously. Conditioned on this, we next proceed to bound $\widetilde{\mathcal{E}}(w)$. Specifically, we have the following holds for all $w\in\mathcal{W}$,
\begin{flalign}
    \widetilde{\mathcal{E}}(w)
    &= \left\Vert \frac{1}{\vert \mathcal{U}_t\vert}\sum_{i\in \mathcal{U}_t}\mathcal{Q}[g_i(w)] - \nabla R_{\mathcal{D}}(w)\right\Vert_2\notag\\
    &= \frac{1}{\vert \mathcal{U}_t\vert}\Bigg\Vert \sum_{i\in \mathcal{G}}(\mathcal{Q}[g_i(w)]
	- \nabla R_{\mathcal{D}}(w))\notag\\
	&\quad-  \sum_{i\in \mathcal{G}\cap\mathcal{T}_t}(\mathcal{Q}[g_i(w_t)] - \nabla R_{\mathcal{D}}(w))\notag\\
	&\quad+ \sum_{i\in \mathcal{B}\cap\mathcal{U}_t}(\mathcal{Q}[g_i(w)] - \nabla R_{\mathcal{D}}(w))\Bigg\Vert_2\notag\\ 
	&\le \frac{1}{\vert \mathcal{U}_t\vert}\left\Vert \sum_{i\in \mathcal{G}}(\mathcal{Q}[g_i(w)] - \nabla R_{\mathcal{D}}(w))\right\Vert_2\notag\\
	&\quad+\frac{1}{\vert \mathcal{U}_t\vert}\left\Vert \sum_{i\in \mathcal{G}\cap\mathcal{T}_t}(\mathcal{Q}[g_i(w)] - \nabla R_{\mathcal{D}}(w))\right\Vert_2\notag\\
	&\quad+\frac{1}{\vert \mathcal{U}_t\vert}\left\Vert \sum_{i\in \mathcal{B}\cap\mathcal{U}_t}(\mathcal{Q}[g_i(w)] - \nabla R_{\mathcal{D}}(w))\right\Vert_2.\label{inq-E}
\end{flalign}
We control each term in (\ref{inq-E}) separately. For the first term in (\ref{inq-E}), we have
\begin{flalign}
    &\frac{1}{\vert \mathcal{U}_t\vert}\left\Vert \sum_{i\in \mathcal{G}}(\mathcal{Q}[g_i(w)] - \nabla R_{\mathcal{D}}(w))\right\Vert_2\notag\\  
    &\le  \frac{1}{\vert \mathcal{U}_t\vert}\left\Vert \sum_{i\in \mathcal{G}}(\mathcal{Q}[g_i(w)]-g_i(w))\right\Vert_2\notag\\
    &\quad+ \frac{1}{\vert \mathcal{U}_t\vert}\left\Vert \sum_{i\in \mathcal{G}}(g_i(w)-\nabla R_{\mathcal{D}}(w))\right\Vert_2\notag\\
	&\le\frac{1}{|\mathcal{U}_t|}\sum_{i\in\mathcal{G}}(\left\|Q[g_i(w)]-g_i(w)\right\|_2)+\frac{1-\alpha}{1-\beta}\mathcal{E}_2\notag\\
	&\le\frac{1}{|\mathcal{U}_t|}\sum_{i\in\mathcal{G}}(\sqrt{1-\delta}\|g_i(w)\|_2)+\frac{1-\alpha}{1-\beta}\mathcal{E}_2\notag\\
	&\le\frac{\sqrt{1-\delta}}{|\mathcal{U}_t|}\sum_{i\in\mathcal{G}}(\|\nabla R_{\mathcal{D}}(w)\|_2\notag\\
	&\quad+\|g_i(w)-\nabla R_{\mathcal{D}}(w)\|_2)+\frac{1-\alpha}{1-\beta}\mathcal{E}_2\notag\\
	&\le\frac{\sqrt{1-\delta}(1-\alpha)}{1-\beta}G+\frac{\sqrt{1-\delta}(1-\alpha)}{1-\beta}\mathcal{E}_1+\frac{1-\alpha}{1-\beta}\mathcal{E}_2\label{eq-T1}
\end{flalign}

Similarly, we bound the second term in (\ref{inq-E}) as follows.
\begin{flalign}
    &\frac{1}{\vert \mathcal{U}_t\vert}\left\Vert \sum_{i\in \mathcal{G}\cap\mathcal{T}_t}(\mathcal{Q}[g_i(w)] - \nabla R_{\mathcal{D}}(w))\right\Vert_2\notag\\ 
    &\le \frac{|\mathcal{T}_t|}{|\mathcal{U}_t|}\max_{i \in \mathcal{G}}\left\Vert \mathcal{Q}[g_i(w)]-\nabla R_{\mathcal{D}}(w) \right\Vert_2\notag\\
	&\le \frac{\beta}{1-\beta}\max_{i \in \mathcal{G}}(\left\Vert \mathcal{Q}[g_i(w)]-g_i(w)\Vert_2+\Vert g_i(w)-\nabla R_{\mathcal{D}}(w) \right\Vert_2)\notag\\
	&\le\frac{\beta}{1-\beta}\max_{i\in\mathcal{G}}(\sqrt{1-\delta}\|g_i(w)\|_2+\Vert g_i(w)-\nabla R_{\mathcal{D}}(w) \Vert_2)\notag\\
	&\le\frac{\beta}{1-\beta}(\sqrt{1-\delta}\|\nabla R_{\mathcal{D}}(w)\|_2\notag\\
	&\quad+(1+\sqrt{1-\delta})\max_{i\in\mathcal{G}}\|g_i(w)-\nabla R_{\mathcal{D}}(w)\|_2)\notag\\
	&\le \frac{\beta\sqrt{1-\delta}}{1-\beta}G+\frac{\beta(1+\sqrt{1-\delta})}{1-\beta}\mathcal{E}_1.\label{eq-T2}
\end{flalign}

Finally, we work on the third term in (\ref{inq-E}). Owing to the trimming threshold $\beta>\alpha$, we have at least one good device machine in the set $\mathcal{T}_t$ for all $t\in[T]$.
	
\begin{flalign}
	&\frac{1}{\vert \mathcal{U}_t\vert}\left\Vert \sum_{i\in \mathcal{B}\cap\mathcal{U}_t}(\mathcal{Q}[g_i(w)] - \nabla R_{\mathcal{D}}(w))\right\Vert_2\notag\\
	&\le \frac{|\mathcal{B}|}{\vert\mathcal{U}_t\vert}\max_{i\in \mathcal{B}\cap\mathcal{U}_t} \Vert \mathcal{Q}[g_i(w)] - \nabla R_{\mathcal{D}}(w) \Vert_2\notag\\
	&\le \frac{\alpha}{1-\beta} \max_{i\in \mathcal{B}\cap\mathcal{U}_t} \left(\Vert \mathcal{Q}[g_i(w)]\Vert_2 + \Vert \nabla R_{\mathcal{D}}(w) \Vert_2\right)\notag\\
	&\le\frac{\alpha}{1-\beta}\max_{i\in \mathcal{B}\cap\mathcal{U}_t} \left(\sqrt{1-\delta}\Vert g_i(w)\Vert_2 + \|g_i(w)\|_2 + \Vert \nabla R_{\mathcal{D}}(w) \Vert_2\right)\notag\\
	&\le \frac{\alpha}{1-\beta}((1+\sqrt{1-\delta})\mathcal{E}_1+(2+\sqrt{1-\delta})\|\nabla R_{\mathcal{D}}(w)\|_2)\notag\\
	&\le\frac{\alpha(2+\sqrt{1-\delta})}{1-\beta}G+\frac{\alpha(1+\sqrt{1-\delta})}{1-\beta}\mathcal{E}_1.\label{eq-T3}
\end{flalign}
	
Combining (\ref{eq-T1}), (\ref{eq-T2}), (\ref{eq-T1}) and letting $\widetilde{\mathcal{E}}$ be the uniform upper bound on $\widetilde{\mathcal{E}}(w)$ over $w\in\mathcal{W}$, we obtain that with probability at least $1-2(\Delta\sqrt{mn})^dd\zeta$, we have the following holds for all $w\in\mathcal{W}$,
\begin{equation}
    \begin{split}
        \widetilde{\mathcal{E}}
        &\le\frac{(1+\beta)\sqrt{1-\delta}+2\alpha}{1-\beta}G\notag\\
        &\quad+\frac{(1+\beta)\sqrt{1-\delta}+\alpha+\beta}{1-\beta}\mathcal{E}_1+\frac{1-\alpha}{1-\beta}\mathcal{E}_2.
    \end{split}
\end{equation}

Note that $\zeta=\frac{1}{2(\Delta\sqrt{mn})^dd(mn)^d}$, hence Lemma~\ref{lm-T_0} follows.

With Lemma~\ref{lm-T_0}, we can prove Theorem~\ref{Thm-hsco}, \ref{Thm-hgc} and \ref{Thm-com-non} by replacing $\mathcal{E}$ with $\widetilde{\mathcal{E}}$ and then following almost the same arguments as in the proof of Theorem~\ref{th-hsco},~\ref{th-hgc}~and~\ref{th-hnc}. Thus the details are omitted here.

\section{Experiments}\label{sec:exp}

We now conduct experiments on both synthetic and real-world data to validate the efficacy of our algorithms. 

\subsection{Experiment Setup}

We will study tasks of linear regression and logistic regression on both synthetic and real-world datasets. The synthetic data is generated as follows. For linear model, we generate each data point ($x$, $y$) by $y=\langle x,w^* \rangle+\xi$, where $\xi\in\mathbb{R}$ is zero-mean noise sampled from $\operatorname{LogNormal}(0, 0.55848)$ by default
and each coordinate of $x\in\mathbb{R}^d$ is sampled from $\operatorname{LogNormal}(0, 0.78)$ by default. 
For logistic model, we generate each data point ($x$,$y$) by $y=\operatorname{sign}[\operatorname{sigmoid}(z)-\frac{1}{2}]$, where $\operatorname{sigmoid}(z)=\frac{1}{1+\exp(-z)}$ and $z=\langle x,w^* \rangle+\xi$ where $\xi\in\mathbb{R}$ is zero-mean noise sampled from $\operatorname{LogNormal}(0, 0.55848)$ by default and each coordinate of $x\in\mathbb{R}^d$ is sample from $\operatorname{LogNormal}(0, 3)$ by default. 
For real-world datasets, we will use the Adult dataset~\cite{Dua:2019} for a binary classification task and the Boston Housing Price dataset\footnote{http://lib.stat.cmu.edu/datasets/boston} for a linear regression task. We use these datasets since they are representative and have been used in previous works on machine learning with heavy-tailed data e.g., \cite{WangX_ICML20}. In particular, to test our algorithms for stochastic non-convex optimization, we train a fully-connected neural network with one hidden layer to solve the tasks on real-world datasets.

Since the precise evaluation of the population risk, i.e., $R_{\mathcal{D}}(w)-R_{\mathcal{D}}(w^*)$, is impossible, we will use the empirical risk to approximate the population risk. Specifically, we measure the performance of different algorithm in terms of test loss, i.e., excess empirical risk on the test dataset. For all experiments, we run each algorithm for at least 10 times and report the average results over all the repetitions.

\subsection{Results and Discussion}

\subsubsection{Comparison with Baseline Methods}

\begin{figure}[tb]
    \centering
    \includegraphics[width=0.5\textwidth]{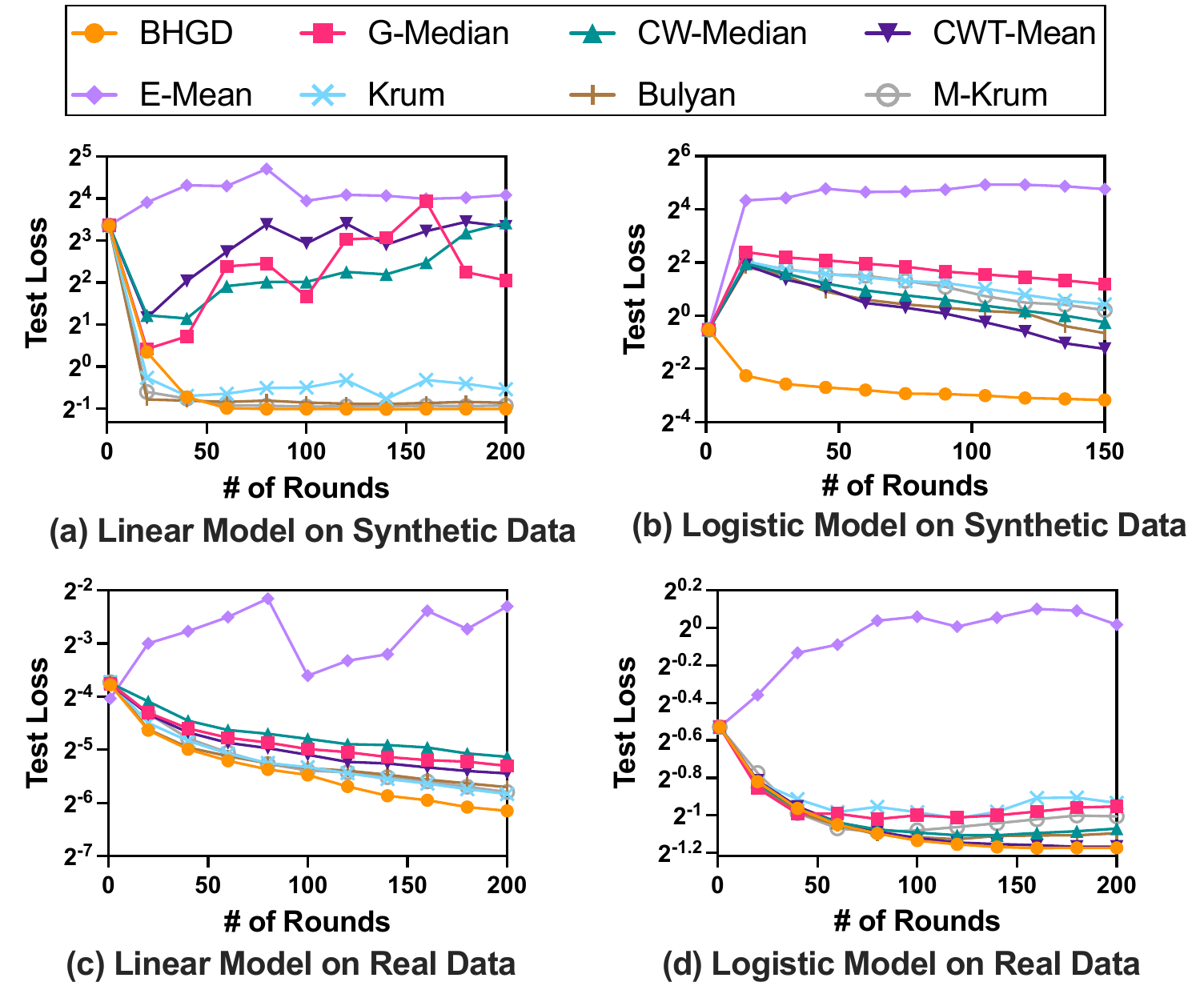}
    \caption{Comparison with Baseline Methods}
    \label{fig-baselines}
\end{figure}

\begin{figure}[tb]
    \centering
    \includegraphics[width=0.5\textwidth]{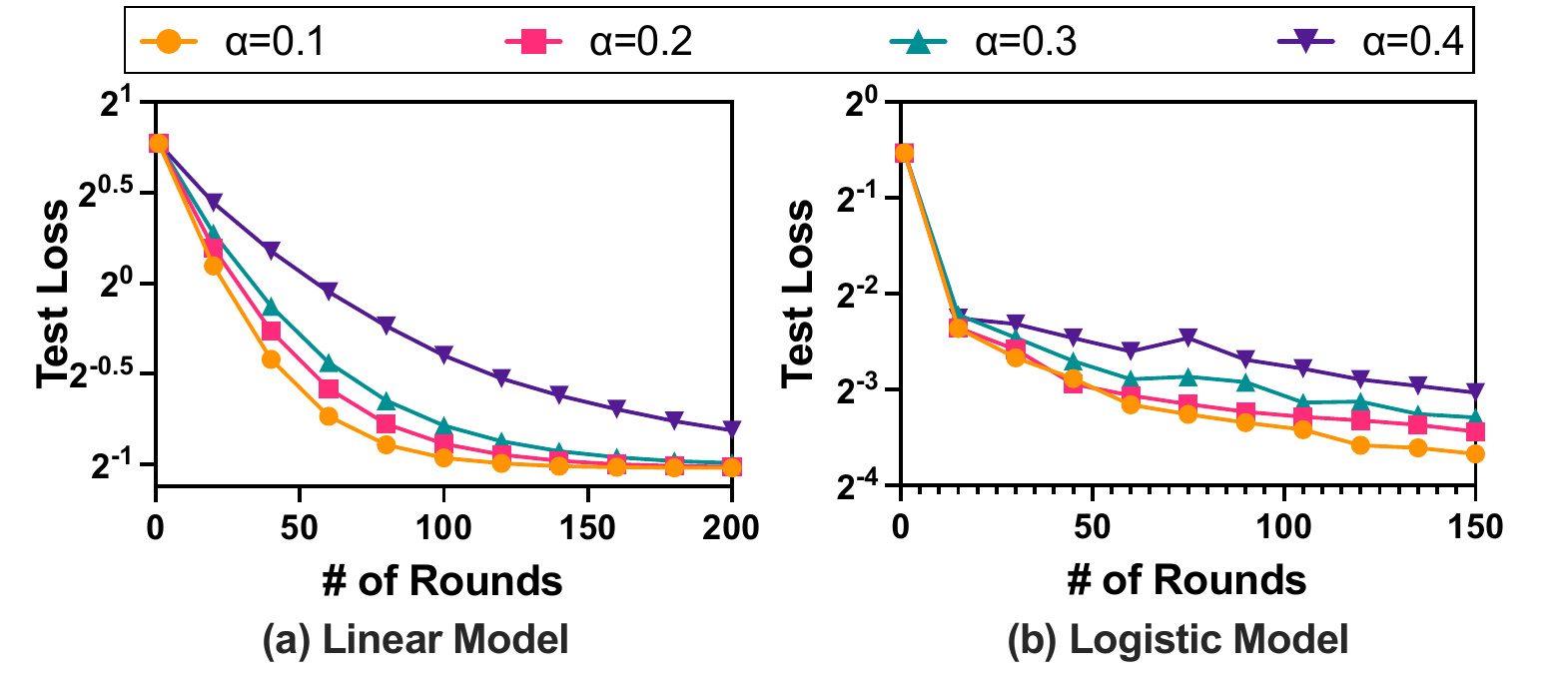}
    \caption{Impact of Byzantine Devices}
    \label{fig-byzantine}
\end{figure}

\begin{figure}[tb]
    \centering
    \includegraphics[width=0.5\textwidth]{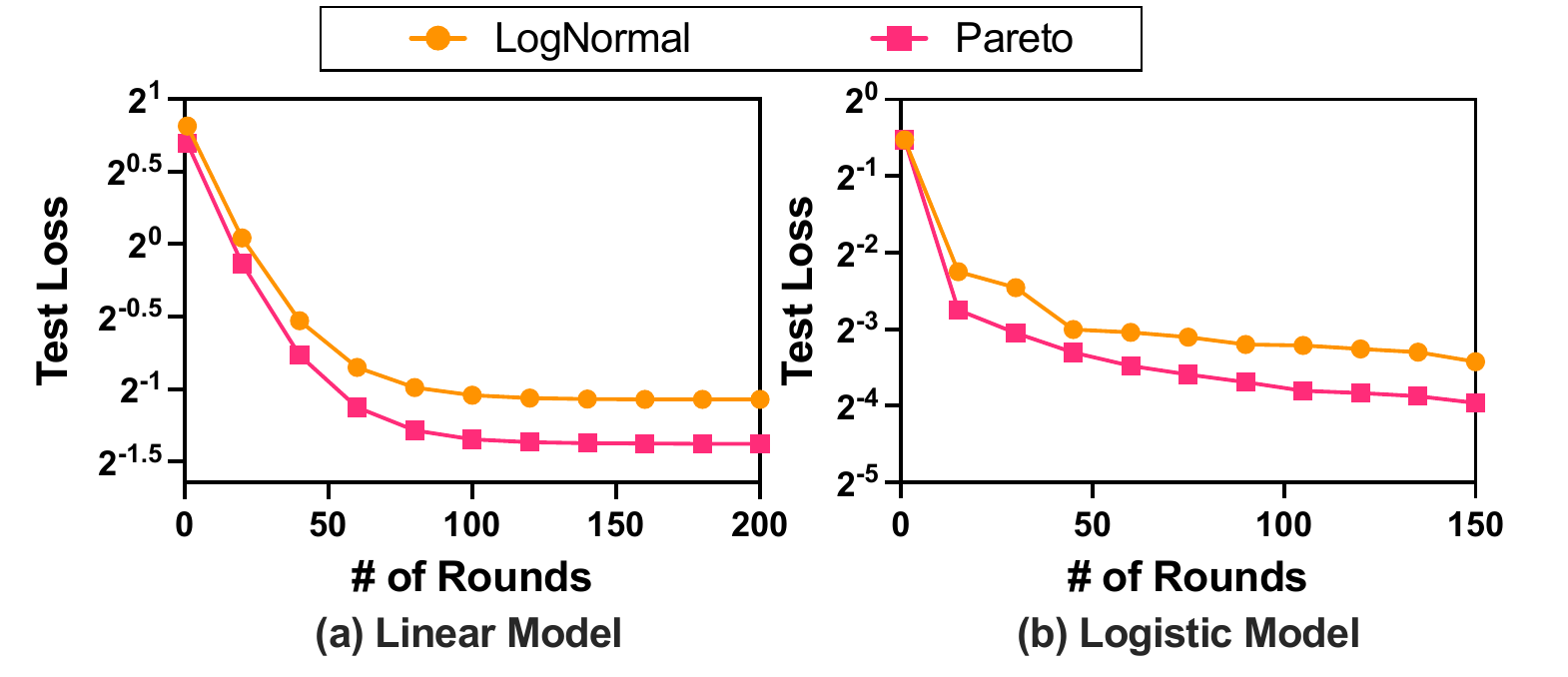}
    \caption{Performance under Different Heavy-tailed Label Noises}
    \label{fig-noises}
\end{figure}

\begin{figure}[tb]
    \centering
    \includegraphics[width=0.5\textwidth]{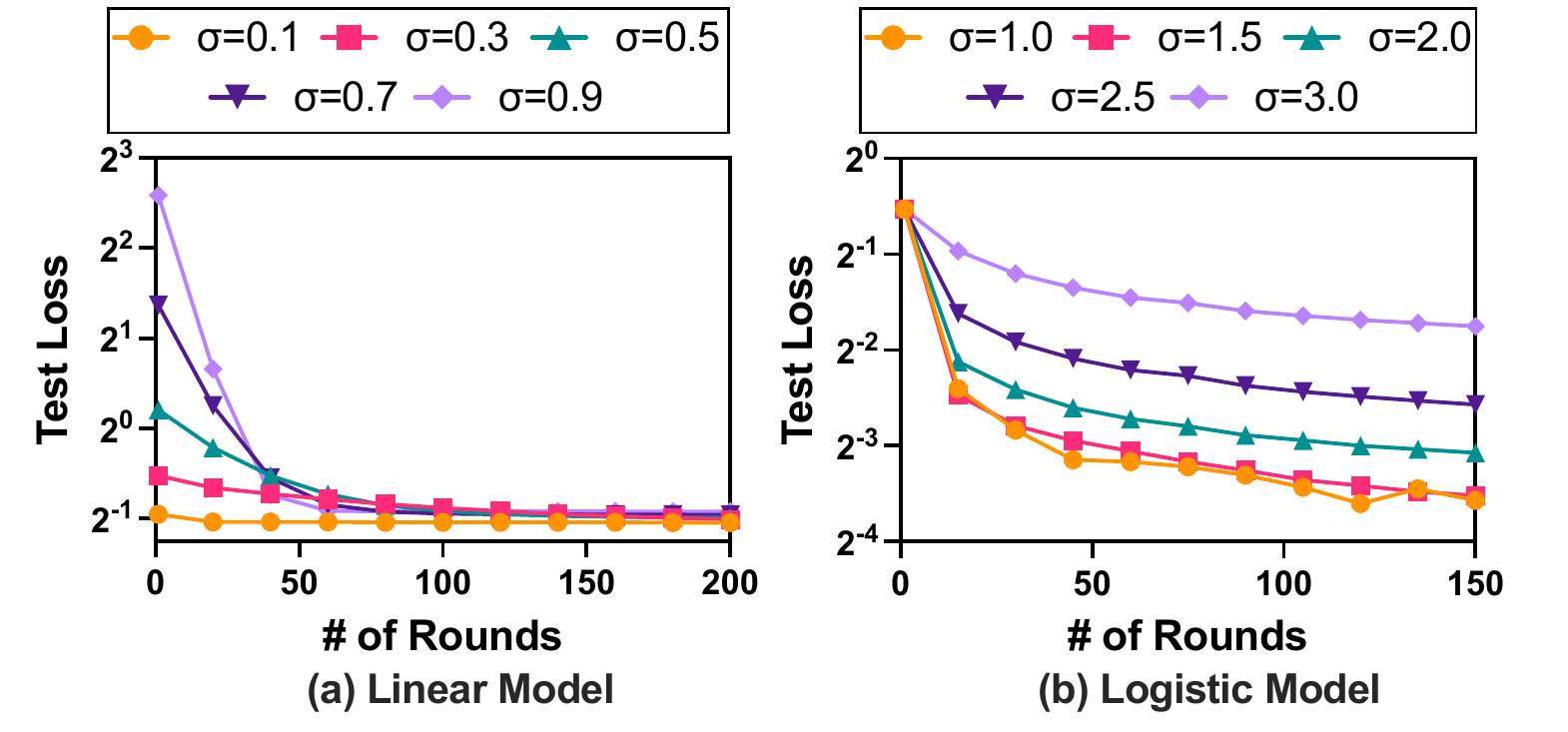}
    \caption{Performance under Different Tail Weights of Heavy-Tailed Feature Vectors}
    \label{fig-tail}
\end{figure}

\begin{figure}[tb]
    \centering
    \includegraphics[width=0.5\textwidth]{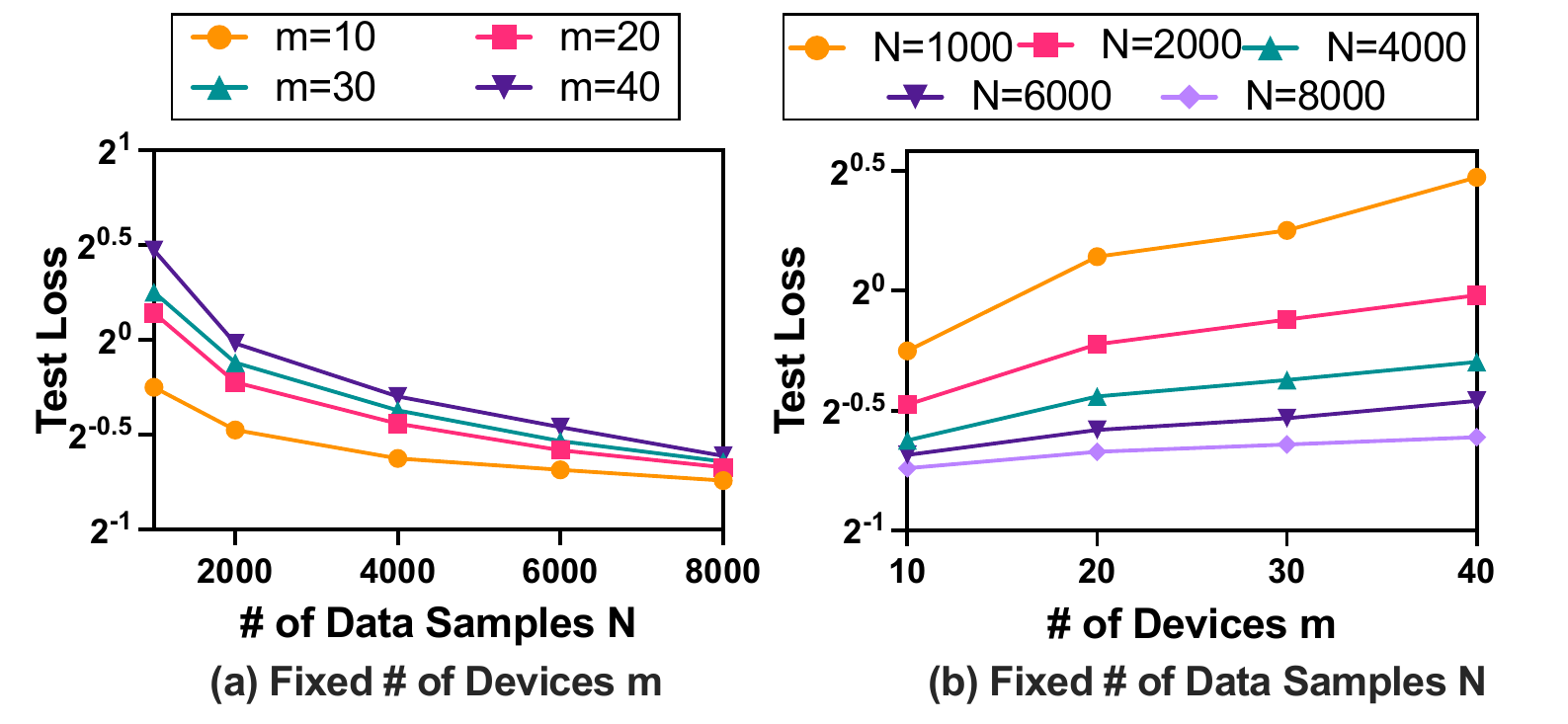}
    \caption{Impact of the System Scale}
    \label{fig-scale}
\end{figure}

\begin{figure}[tb]
    \centering
    \includegraphics[width=0.5\textwidth]{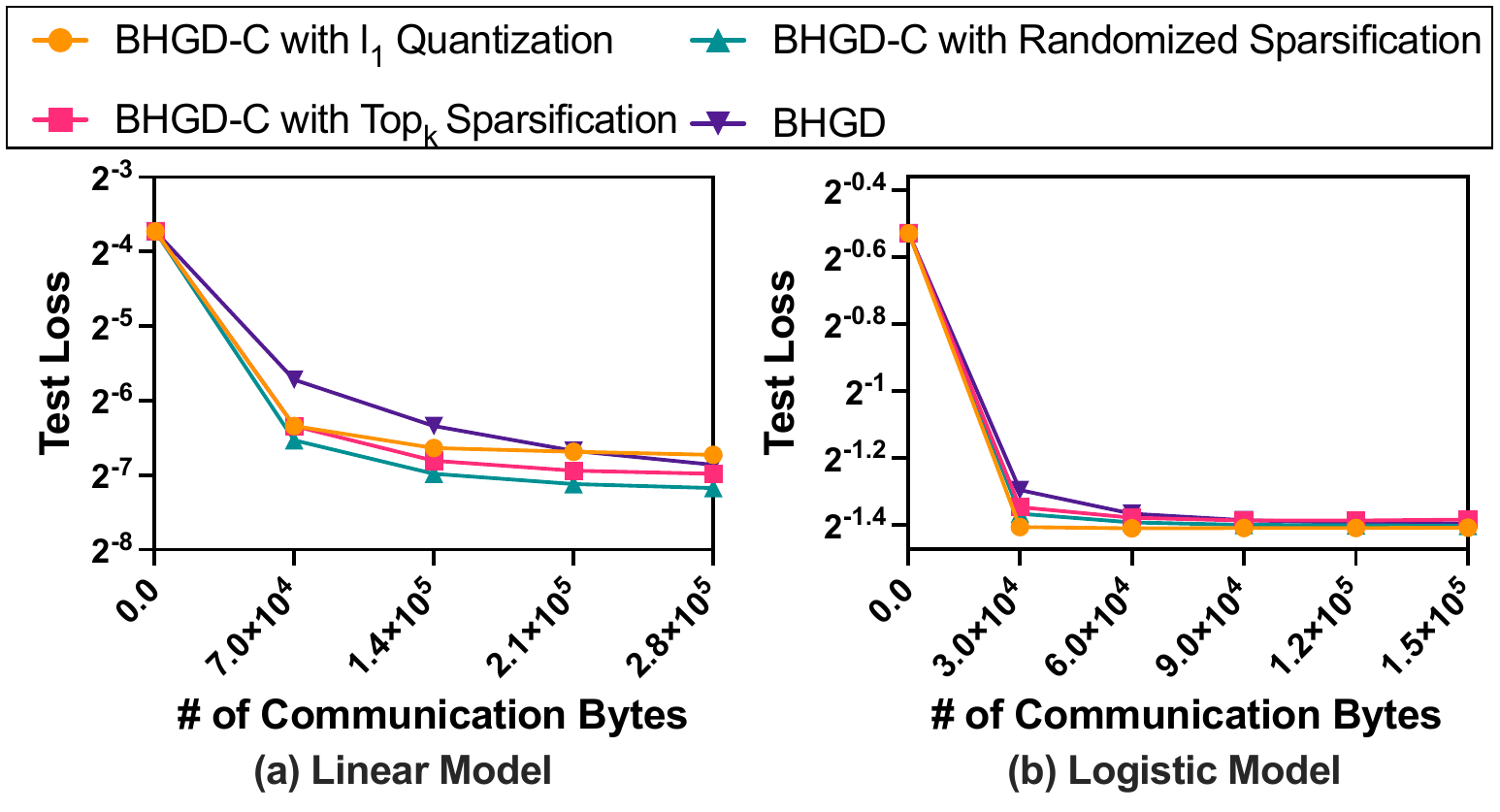}
    \caption{Impact of Gradient Compression}
    \label{fig-compression}
\end{figure}

\begin{figure}[tb]
    \centering
    \includegraphics[width=0.47\textwidth]{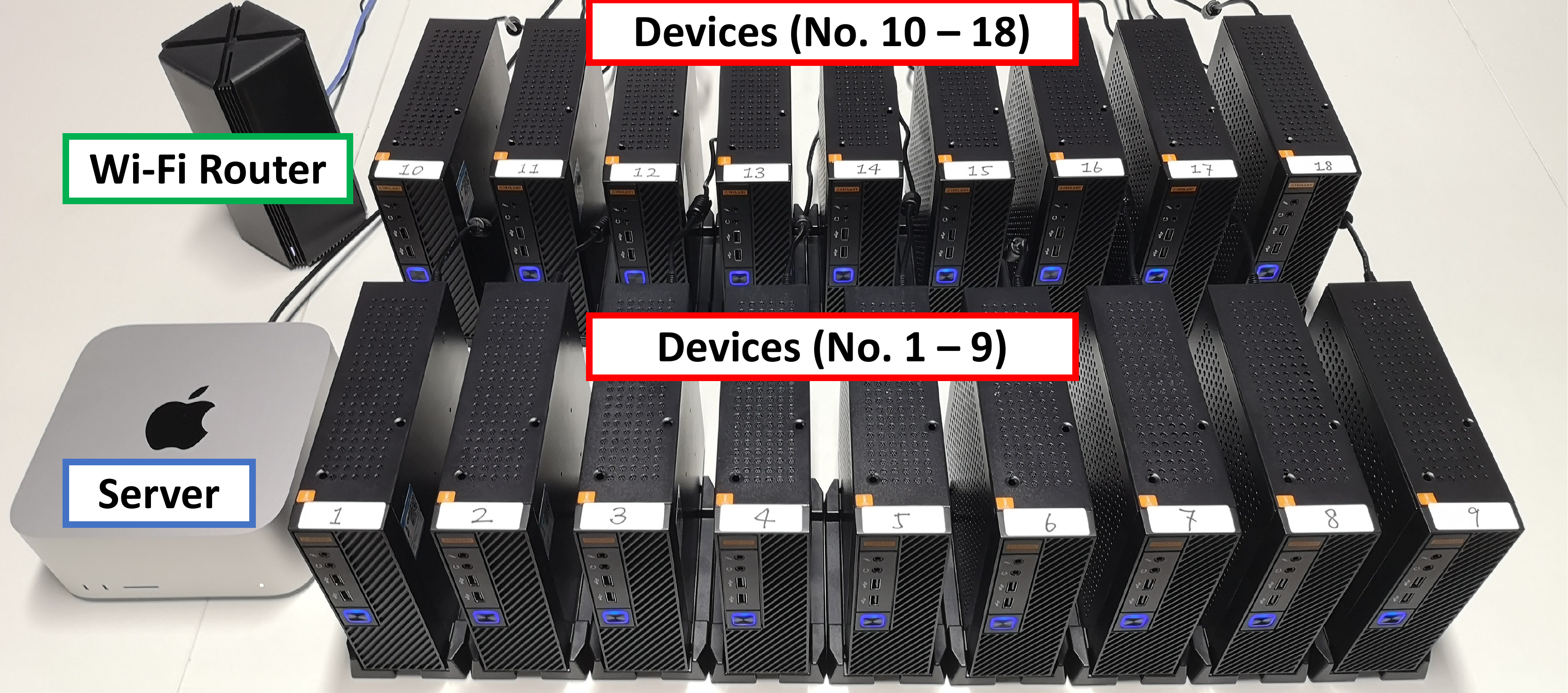}
    \caption{The hardware prototype system for edge federated learning}
    \label{fig-system}
\end{figure}

\begin{figure}[tb]
    \centering
    \includegraphics[width=0.5\textwidth]{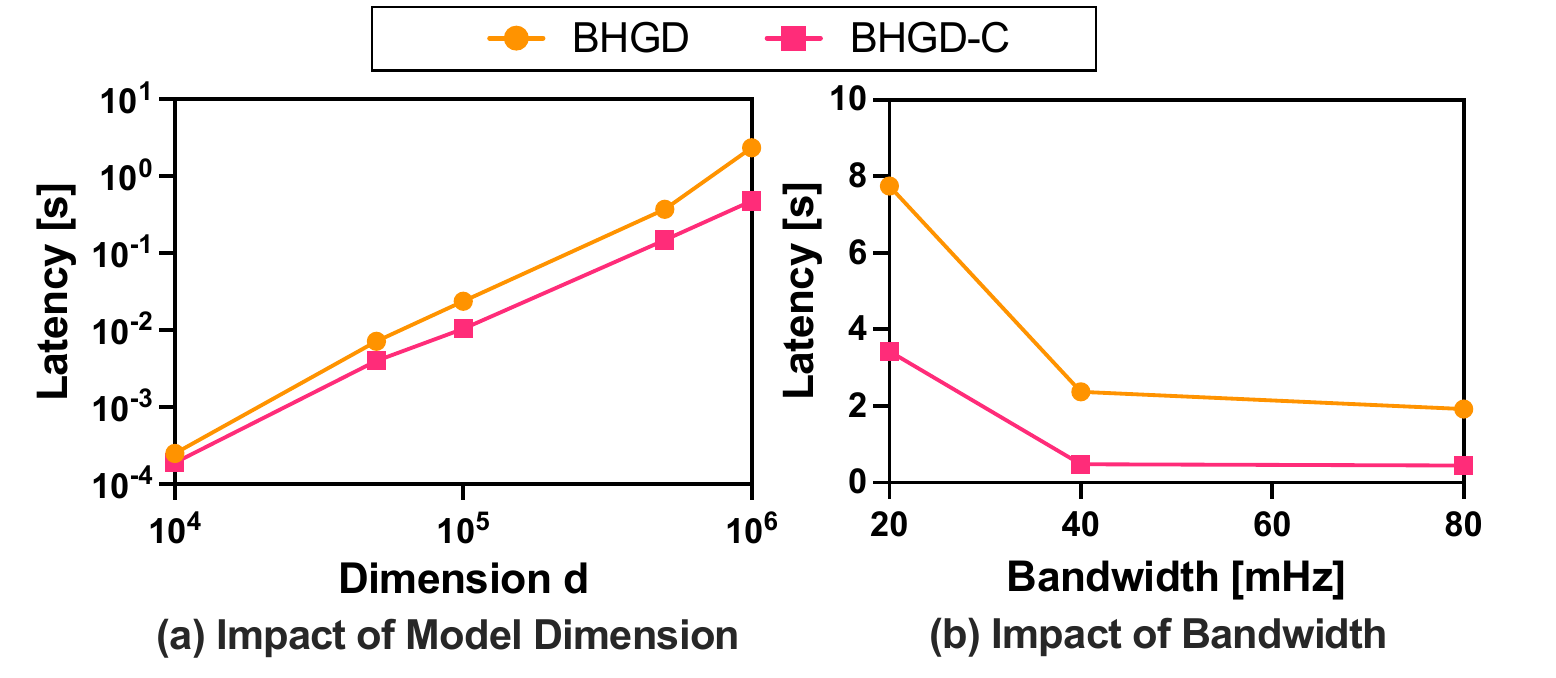}
    \caption{Transmission Latency in the Real System}
    \label{fig-latency}
\end{figure}


To show the Byzantine-resilience of BHGD (Algorithm~\ref{Algo1}), we compare the performance of BHGD with previous methods. Specifically, we consider the distributed gradient descent algorithms with averaging rules being empirical mean (E-Mean)~\cite{McMahanMRHA_AISTATS17}, coordinate-wise trimmed mean (CWT-Mean)~\cite{YinCRB_ICML18}, coordinate-wise median (CW-Median)~\cite{YinCRB_ICML18}, geometric median (G-Median)~\cite{ChenSX_Sigmetrics17}, Krum~\cite{BlanchardMGS_NIPS17}, Bulyan~\cite{guerraoui2018hidden}, and momentum Krum (M-Krum)~\cite{el2021distributed} respectively. For experiments on  synthetic datasets, we set $m=10$, $n=100$, $d=10$, $\alpha=0.2$ and generate $200$ test data for each model. For experiments on real-world datasets, we set $m=10$, $n=40$, $\alpha=0.2$ and $m=5$, $n=10$, $\alpha=0.2$ for Boston dataset and Adult dataset respectively and pick $100$ data uniformly at random as the test data for each. The experimental results are shown in Fig.~\ref{fig-baselines}. As we can see from Fig.~\ref{fig-baselines}, BHGD (Algorithm~\ref{Algo1}) achieves lower test loss than all the baseline methods on both synthetic and real-world datasets. In addition, the convergence of BHGD is stabler than the other methods, which indicates that BHGD is more robust to heavy-tailed data as well as Byzantine devices.

\subsubsection{Impact of Byzantine Devices}

We conduct experiments on synthetic datasets to study the impact of Byzantine devices. We present the performance of BHGD (Algorithm~\ref{Algo1}) in presence of different fractions of Byzantine devices in Fig.~\ref{fig-byzantine}, where $\alpha$ denotes the fraction of Byzantine devices among all learning devices. Our BHGD algorithm converge in all cases, which confirms that BHGD can tolerate various fraction of Byzantine devices. Moreover, the convergence is subject to the fraction of Byzantine devices. As shown in Fig.~\ref{fig-byzantine}, a larger fraction of Byzantine devices may not only cause a larger excess empirical risk, but also slow the convergence down.

\subsubsection{Impact of Heavy-tailed Distributions}

We conduct experiments on synthetic datasets with different types of heavy-tailed label noises and different tail weights of heavy-tailed feature vectors to show the impact of heavy-tailed distributions on the learning performance of BHGD algorithm. First, we adopt LogNormal distribution and Pareto Distribution to generate the label noises. The main difference between these two heavy-tailed distributions is that LogNormal distribution is symmetric while Pareto distribution is one-sided. For Pareto label noise in each task, we set the scale parameter to be $1$ and the shape parameter to be $3.26953$ such that its second moment is identical to that of $\operatorname{LogNormal}(0, 0.55848)$. The experimental results are shown in Fig.~\ref{fig-noises}. It can be seen that, BHGD converges under both types of noises, and performs better on Pareto noises. We think the main reason lies in that Pareto noise is one-sided and thus its influence is smaller. Next, we vary the tail weight of heavy-tailed feature vector. We adopt the LogNormal distribution $\operatorname{LogNormal}(\mu, \sigma)$, where $\mu$ is set to be $0$ and $\sigma$ varies in different settings. Note that, a larger $\sigma$ corresponds to a larger variance. For linear model, we let $\sigma$ vary in $\{0.1, 0.3, 0.5, 0.7, 0.9\}$, while for logistic model, we let $\sigma$ vary in $\{1.0, 1.5, 2.0, 2.5, 3.0\}$. The experimental results are shown in Fig.~\ref{fig-tail}. From these results, we can conclude that more test loss would be incurred when the training data become more heavy-tailed.

\subsubsection{Impact of System Scale}

We study the impact of system scale on the performance of BHGD. Specifically, we consider the scale of learning devices and the scale of training data samples. In this experiment, we let $m$ (i.e., the number of learning devices) vary in $\{10, 20, 30, 40\}$ and let $N$ (i.e., the number of training data samples) vary in $\{1000, 2000, 4000, 6000, 8000\}$. We run the BHGD algorihtm on synthetic dataset for logistic model under each combination of $m$ and $N$. The results are shown in Fig.~\ref{fig-scale}. In Fig.~\ref{fig-scale}~(a), we fix $m$ and present the test loss after $200$ training rounds for various $N$. Since $N=mn$, when $m$ is fixed, a larger $N$ means a larger size of local data sample $n$ for each device. The results in Fig.~\ref{fig-scale}~(a) shows that when each learning device has more training data, the test loss would get smaller. Intuitively, this means that when each device has more local data, it learns more accurately, which is quite reasonable. In Fig.~\ref{fig-scale}~(b), we fix $N$ and present the test loss after $200$ training rounds for various $m$.When $N$ is fixed, a larger $m$ means a smaller size of local data $n$ for each device. Hence the test loss grows larger as $m$ increases in Fig.~\ref{fig-scale}~(b). Note that, these results corroborate our previous theoretical upper bounds on statistical error rates.

\subsubsection{Communication Efficiency}

We show the communication efficiency of our BHGD-C algorithm (Algorithm~\ref{Algo2}). 

In the first part, we study the impact of different gradient compression techniques. To this end, we equip BHGD-C with different compressor, and then compare their performances with BHGD on real-world datasets. Specifically, we adopt $\ell_1$ quantization~\cite{KarimireddyRSJ_ICML19}, $\text{Top}_k$ sparsification~\cite{StichCJ_NIPS18} and randomized sparsification~\cite{wangni2018gradient}.  In $\ell_1$ quantization, the compressor $Q(\cdot)$ can be summarized as $\mathcal{Q}(x)=\left\{\frac{\|x\|_1}{d}, \operatorname{sign}(x)\right\}$ for $\forall x\in\mathbb{R}^d$, where $\operatorname{sign}(x)$ is the quantized vector with each coordinate $i\in[d]$ being either $+1$ (for positive $x_i$) or $-1$ (for negative $x_i$) and $\frac{\|x\|_1}{d}$ is the scaling factor. In $\text{Top}_k$ sparsification, for any $x\in\mathbb{R}^d$, the compressor compresses $x$ by retaining the top $k$ largest
coordinates of $x$ and sets the others to zero. In randomized sparsification, for any $x\in\mathbb{R}^d$ and each coordinate $i\in[d]$, the compressor set $x_i$ to $1$ with probability $p$ and to $0$ with probability $1-p$. The experimental results are shown in Fig.~\ref{fig-compression}. Note that, in Fig.~\ref{fig-compression}, we present how the test loss varies with respect to the total communication costs (in byte) during the training process. It can be seen that, to achieve the same test loss, BHGD-C entails fewer bytes to transmit from devices to the sever and thus provide communication efficiency. 

Next, in the second part, we conduct experiment on a networked hardware prototype system for edge federated learning. Our system, as illustrated in Fig.~\ref{fig-system}, consists of $m=18$ Teclast mini PCs serving as learning devices and a Mac Studio serving as the central server. All the machines are interconnected via a Wi-Fi router and the bandwidth is set to 40mHz by default. We implement our BHGD and BHGD-C algorithms on the system, where BHGD-C algorithm use the $\text{Top}_k$ sparsification with $k$ being half of the model dimension. We present the average communication latency across all the learning devices and all the training rounds in Fig.~\ref{fig-latency}. Overall, the BHGD-C incurs less latency during the learning process. From Fig.~\ref{fig-latency}~(a), we can see that, with the model dimension increases, the latency saved by BHGD-C also grows. From Fig.~\ref{fig-latency}~(b), we can see that, the bandwidth affects the communication latency a lot and BHGD-C achieves a much lower latency than BHGD even when the bandwidth is limited.

\section{Conclusions}\label{sec:conclu}

In this paper, we studied how to retain Byzantine resilience and communication efficiency when training machine learning model with heavy-tailed data in a distributed manner. Specifically, we first presented an algorithm that is robust against both Byzantine worker nodes and heavy-tailed data. Then by adopting the gradient compression technique, we further proposed a robust learning algorithm with reduced communication overhead. Our theoretical analysis demonstrated that our proposed algorithms achieve optimal error rates for strongly convex population risk case. We also conducted extensive experiments, which yield consistent results with the theoretical analysis. It is worth noting that, our results may not have optimal dependence on the dimension $d$. Some recent work also attempt to incorporate recent breakthroughs in robust high-dimensional aggregators into the distributed learning scenarios, e.g., \cite{DiakonikolasKK0_ICML19,SuX_Sigmetrics19,El-MhandiG_arxiv19,DataD_ICML21}. How to achieve Byzantine resilience, communication efficiency and heavy-tailed data robustness simultaneously in high dimensions is an interesting and significant problem and we leave it as the future work.
\section*{Acknowledgments}
This work was supported in part by the National Key Research and Development Program of China Grant 2020YFB1005900 and National Natural Science Foundation of China (NSFC) Grant 62122042.


\ifCLASSOPTIONcaptionsoff
  \newpage
\fi



\normalem
\bibliographystyle{IEEEtran}
 \bibliography{references}
%



%


\begin{IEEEbiography}[{\includegraphics[width=1in,height=1.25in,clip,keepaspectratio]{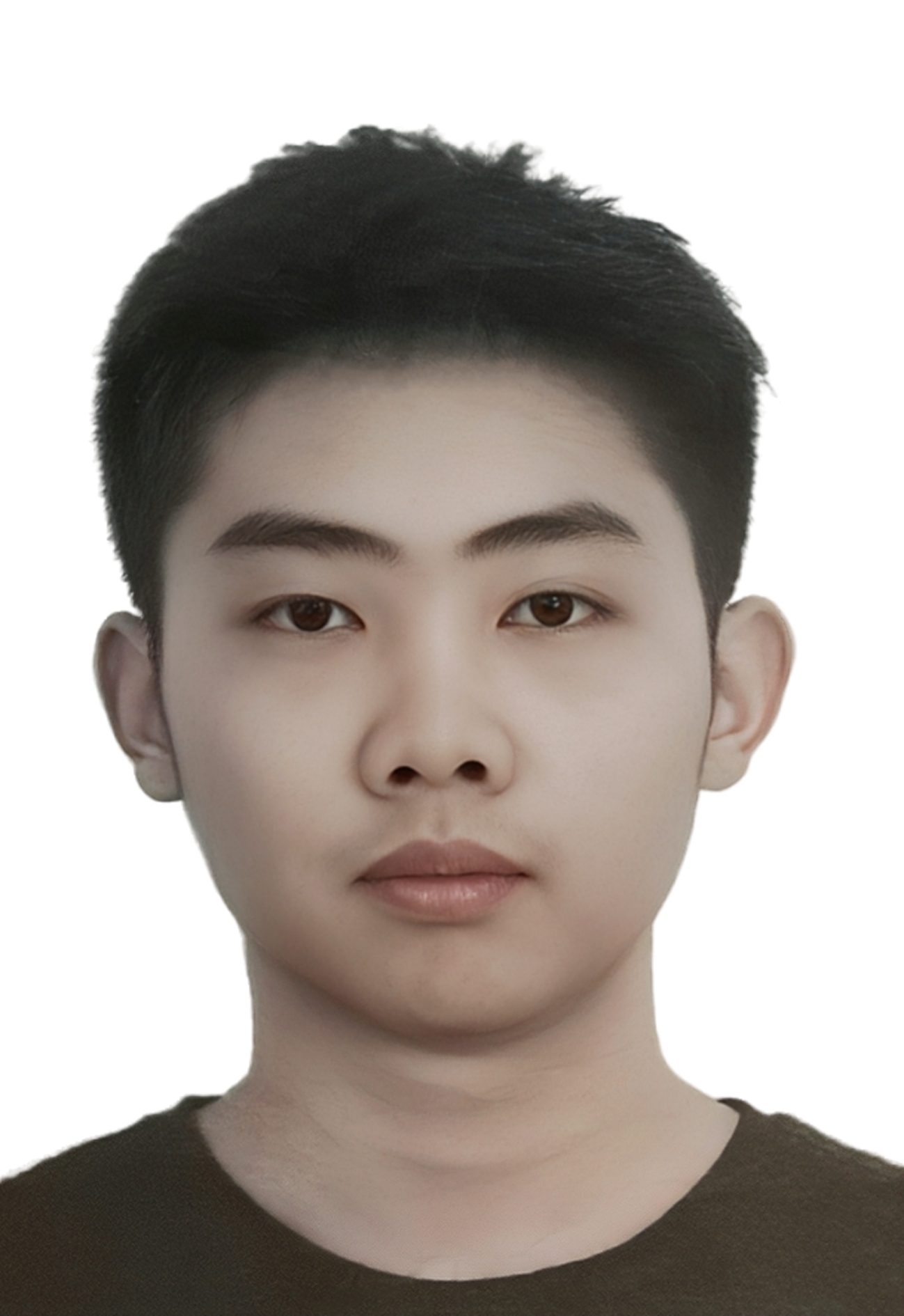}}]{Youming Tao}
received the Bachelor's degree (with distinction) in Computer Science from Shandong University in 2021, where he is currently pursuing the Ph.D. degree in Computer Science. He is interested in building private, efficient and reliable collaborative learning algorithms with applications to edge intelligence.
\end{IEEEbiography}


\begin{IEEEbiography}[{\includegraphics[width=1in,height=1.25in,clip,keepaspectratio]{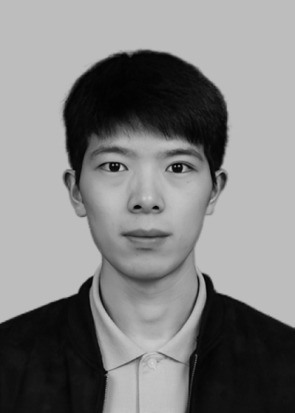}}]{Sijia Cui}
is currently a Ph.D. student in Computer Science at University of Chinese Academy of Sciences. He obtained his B.E. degree in computer science from Shandong University in 2022. He is interested in reinforcement learning.
\end{IEEEbiography}

\begin{IEEEbiography}[{\includegraphics[width=1in,height=1.25in,clip,keepaspectratio]{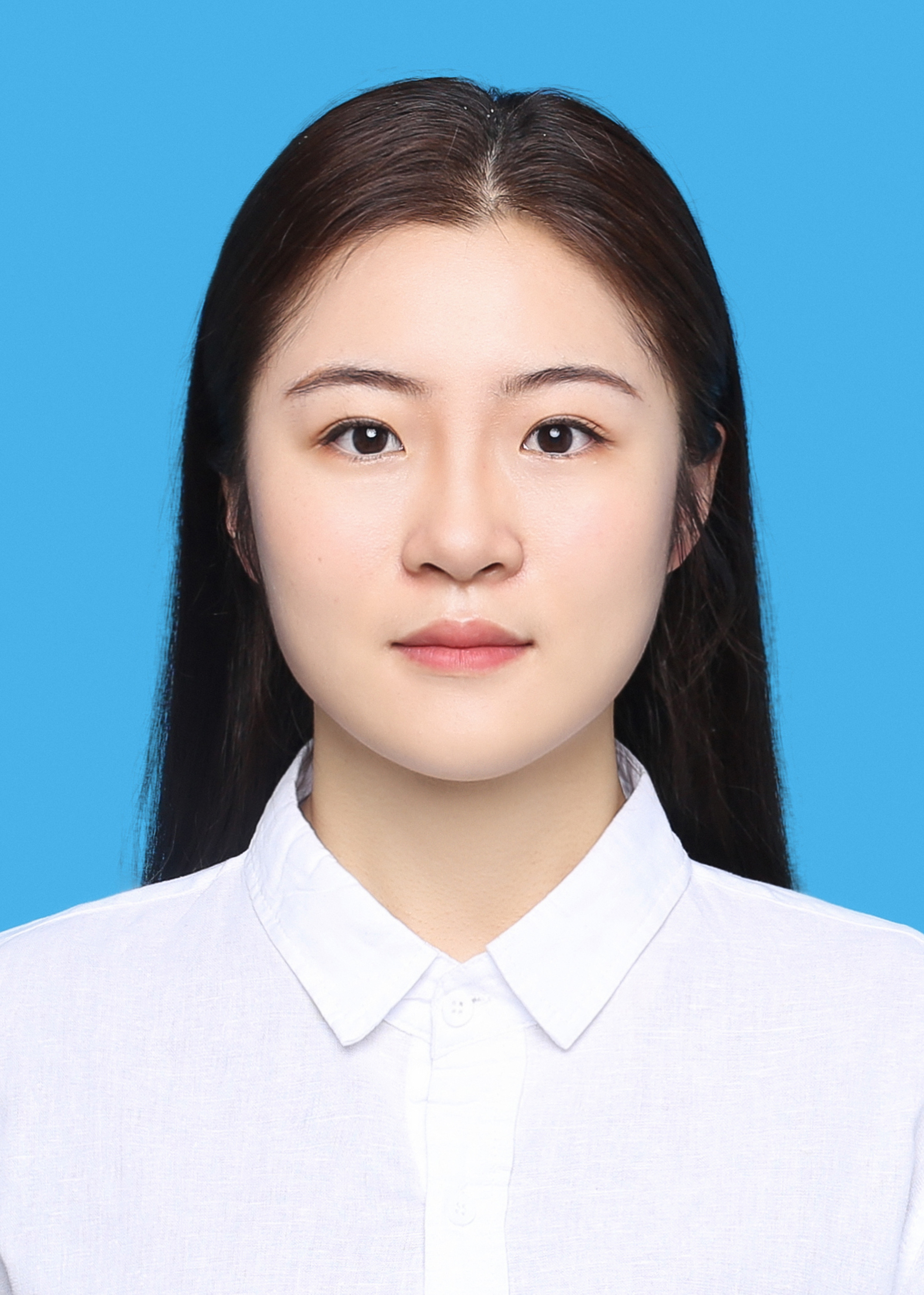}}]{Wenlu Xu}
is currently a Ph.D. student at the Department of Statistics, University of California, Los Angeles (UCLA). She received her bachelor’s degree in Mathematics from Shandong University. She is under the guidance of Dr. Xiaowu Dai. Her research interests include recommender system, causal inference and machine learning in economics.  She is currently a member of the Lab for Statistics, Computing, Algorithm, Learning, and Economics (SCALE).
\end{IEEEbiography}

\begin{IEEEbiography}[{\includegraphics[width=1in,height=1.25in,clip,keepaspectratio]{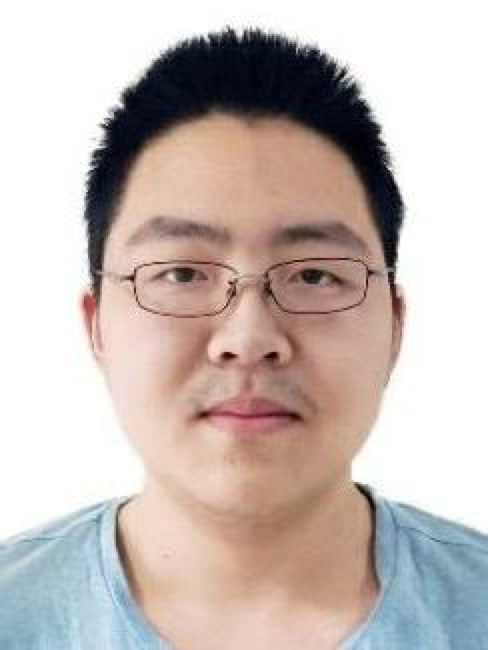}}]{Haofei Yin}
is currently a Master student in computer science at Shandong University. He obtained his B.E. degree in Computer Science from Shandong University in 2022. He is interested in computer networking.
\end{IEEEbiography}

\begin{IEEEbiography}[{\includegraphics[width=1in,height=1.25in,clip,keepaspectratio]{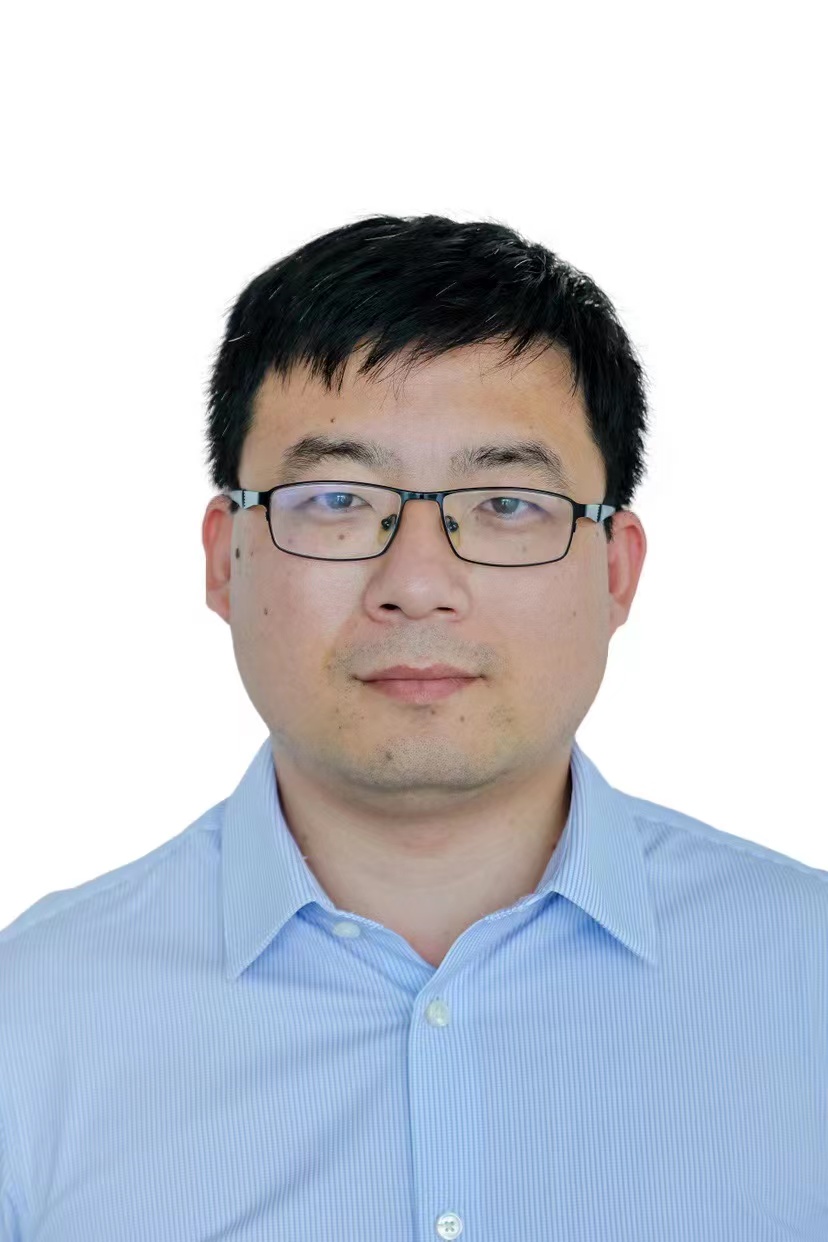}}]{Dongxiao Yu}
received the B.Sc. degree from the School of Mathematics, Shandong University, in 2006, and the Ph.D. degree from the Department of Computer Science, The
University of Hong Kong, in 2014. He became an Associate Professor with the School of Computer Science and Technology, Huazhong University of Science and Technology, in 2016. He is currently a Professor with the School of Computer Science and Technology, Shandong University. His research interests include wireless networks, distributed
computing, and graph algorithms.
\end{IEEEbiography}

\begin{IEEEbiography}[{\includegraphics[width=1in,height=1.25in,clip,keepaspectratio]{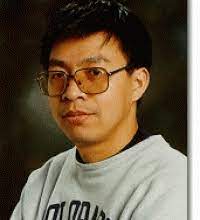}}]{Weifa Liang}
received the B.Sc. degree in computer science from Wuhan University, China, in 1984 and the M.E. and Ph.D. degrees in computer science from Australian National University, in 1989 and 1998, respectively. He is currently a professor with the Department of Computer Science, City University of Hong Kong. Prior to that, he was a professor with Australian National University. His research interests include design and analysis of energy efficient routing protocols for wireless ad hoc and sensor networks, mobile edge computing (MEC), network function virtualization (NFV), Internet of Things, software-defined networking (SDN), design and analysis of parallel and distributed algorithms, approximation algorithms, combinatorial optimization, and graph theory. He is currently an associate editor for the editorial board of IEEE Transactions on Communications.
\end{IEEEbiography}

\begin{IEEEbiography}[{\includegraphics[width=1in,height=1.25in,clip,keepaspectratio]{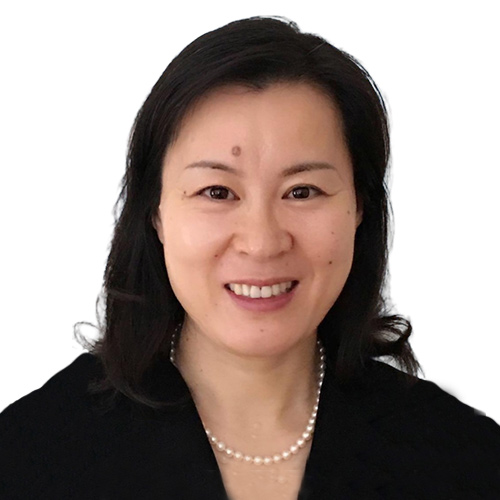}}]{Xiuzhen Cheng}
received her M.S. and Ph.D. degrees in computer science from University of Minnesota, Twin Cities, in 2000 and 2002, respectively. She was a faculty member at the Department of Computer Science, The George Washington University, from 2002-2020. Currently she is a professor of computer science at Shandong University, Qingdao, China. Her research focuses on blockchain computing, security and privacy, and Internet of Things. She is a Fellow of IEEE
\end{IEEEbiography}




\end{document}